\documentclass[a4paper,10pt]{article}
\usepackage[utf8x]{inputenc}

\usepackage{amsmath,amssymb,amsfonts,amsthm,dsfont}

\usepackage{bbm}
\usepackage{hyperref}
\usepackage{graphicx}
\usepackage{epstopdf}
\usepackage{pgfplots}
\usepackage{tikz}
\usetikzlibrary{calc,fixedpointarithmetic}

\newtheorem{thm}{Theorem}[section]
\newtheorem{lem}[thm]{Lemma}
\newtheorem{rem}[thm]{Remark}
\newtheorem{cor}[thm]{Corollary}
\newtheorem{definition}[thm]{Definition}

\let\Im\undefined
\DeclareMathOperator{\Im}{Im}

\DeclareMathOperator{\tr}{tr}
\DeclareMathOperator{\dd}{d\!}
\DeclareMathOperator{\rank}{rank}

\DeclareMathOperator{\Span}{span}
\DeclareMathOperator{\diag}{diag}

\DeclareMathOperator{\Ltwolim}{\mathsf L^2-lim}
\DeclareMathOperator{\sech}{sech}

\numberwithin{equation}{section}

\title{Lower bounds on the moduli of three-dimensional Coulomb-Dirac operators via fractional Laplacians with applications}
\author{Sergey Morozov\footnote{Mathematisches Institut, Ludwig-Maximilians-Universit\"at M\"unchen, Theresienstr. 39, 80333 Munich, Germany\newline $\mathtt{\qquad morozov@math.lmu.de,\ dmueller@math.lmu.de}$}\ \footnote{St. Petersburg State University, Universitetskaya nab. 7-9, 199034 St. Petersburg, Russia} and David M\"uller$^{*}$}
\date{}

\begin{document}

\maketitle

\begin{abstract}
For $\nu\in[0, 1]$ let $D^\nu$ be the distinguished self-adjoint realisation of the three-dimensional Coulomb-Dirac operator $-\mathrm i\boldsymbol\alpha\cdot\nabla -\nu|\cdot|^{-1}$. For $\nu\in[0, 1)$ we prove the lower bound of the form $|D^\nu| \geqslant C_\nu\sqrt{-\Delta}$, where $C_\nu$ is found explicitly and is better then in all previous works on the topic. In the critical case $\nu =1$ we prove that for every $\lambda\in [0, 1)$ there exists $K_\lambda >0$ such that the estimate $|D^{1}| \geqslant K_\lambda a^{\lambda -1}(-\Delta)^{\lambda/2} -a^{-1}$ holds for all $a >0$. As applications we extend the range of coupling constants in the proof of the stability of the relativistic electron-positron field and obtain Cwickel-Lieb-Rozenblum and Lieb-Thirring type estimates on the negative eigenvalues of perturbed projected massless Coulomb-Dirac operators in the Furry picture. We also study the existence of a virtual level at zero for such projected operators.
\end{abstract}

\section{Introduction and main results}
This work is dedicated to the study of the Coulomb-Dirac operator
\begin{equation}\label{Dirac symbol}
 -\mathrm i\boldsymbol\alpha\cdot\nabla +M\beta- \nu|\cdot|^{-1}
\end{equation}
in $\mathsf L^2(\mathbb R^3, \mathbb C^4)$. 
Here $\boldsymbol{\alpha} :=(\alpha_1,\alpha_2,\alpha_3)$ is the vector of $\alpha_{i} :=\begin{pmatrix}
 0_{\mathbb{C}^2} & \sigma_{i} \\ \sigma_{i} & 0_{\mathbb{C}^2}
\end{pmatrix}$ for $i\in\{1,2,3\}$ with $\sigma_1,\sigma_2,\sigma_3$ being the standard Pauli matrices, $\beta:= \diag(1, 1, -1, -1)$ and $M \geqslant0$ is the mass of the Dirac particle. It is well known that the operator \ref{Dirac symbol} is essentially self-adjoint on $\mathsf C_0^\infty\big(\mathbb R^3\setminus\{0\}, \mathbb C^4\big)$ only for $|\nu| \leqslant \sqrt3/2$. For $\nu\in(\sqrt3/2, 1]$ there is a canonical choice of a self-adjoint extension, which for $M =1$ we denote by $D^{\nu, 1}$, see \cite{EstebanLoss,Mueller} and Remark \ref{r: D^nu}. For general $M \geqslant 0$ we let $D^{\nu, M} :=D^{\nu, 1} +(M -1)\beta$. We will restrict our attention to $\nu\in[0, 1]$ and the distinguished self-adjoint realisation $D^{\nu, M}$. A special role is played by the homogeneous operator $D^\nu :=D^{\nu, 0}$.

For $\nu\in [0, 1]$ we introduce
\begin{equation}\label{beta}
 \Upsilon_\nu:= \sqrt{1 -\nu^2}.
\end{equation}

We now state the results of the paper. Some of them are related to the corresponding results in two dimensions, which we have obtained in \cite{MorozovMueller}. It turns out that the three-dimensional situation is somewhat simpler due to presence of the Hardy inequality. The key results are Theorems \ref{absolute value bounds theorem} and \ref{t: l bound}, most of the rest are their applications. It is important to note that we provide the explicit constant $C_{\nu}$ in Theorem \ref{absolute value bounds theorem}, whereas in the corresponding Theorem 1 in \cite{MorozovMueller} substantial work is still required to extract an explicit value from the proof.

In the following scalar operators like $\sqrt{-\Delta}$ are applied to vector-valued functions component-wise without reflecting this in the notation.

\subsection{Lower bounds via powers of the Laplacian}

\begin{thm}\label{absolute value bounds theorem}
For every $\nu\in [0, 1)$ the inequality
\begin{align}\label{absolute value bound}
 |D^\nu| \geqslant C_\nu\sqrt{-\Delta}
\end{align}
holds with
\begin{align}\label{c nu}
C_{\nu}:=\big(1-\pi\Upsilon_\nu\cot(\pi\Upsilon_\nu/2)/2\big)\eta_\nu,
\end{align}
where
\begin{align}\label{eta_nu}
\eta_{\nu}:=\begin{cases}
             \dfrac{(9+4\nu^2)^{1/2}-4\nu}{3\big(1-2\Upsilon_\nu\cot(\pi\Upsilon_\nu/2)\big)},&\text{for }\nu\in [0, 1)\setminus\{\sqrt3/2\};\\
             \dfrac1{\sqrt 3(\pi -2)}, &\text{for } \nu =\sqrt3/2;\\
	     \dfrac{\pi(4 -\sqrt{13})}{3(4 -\pi)}, &\text{for } \nu =1.
            \end{cases}
\end{align}
\end{thm}

The following figure shows the graph of $C_\nu$ as a function of $\nu$ with the dotted affine function for reference.
\begin{center}
\includegraphics[width=0.7\textwidth]{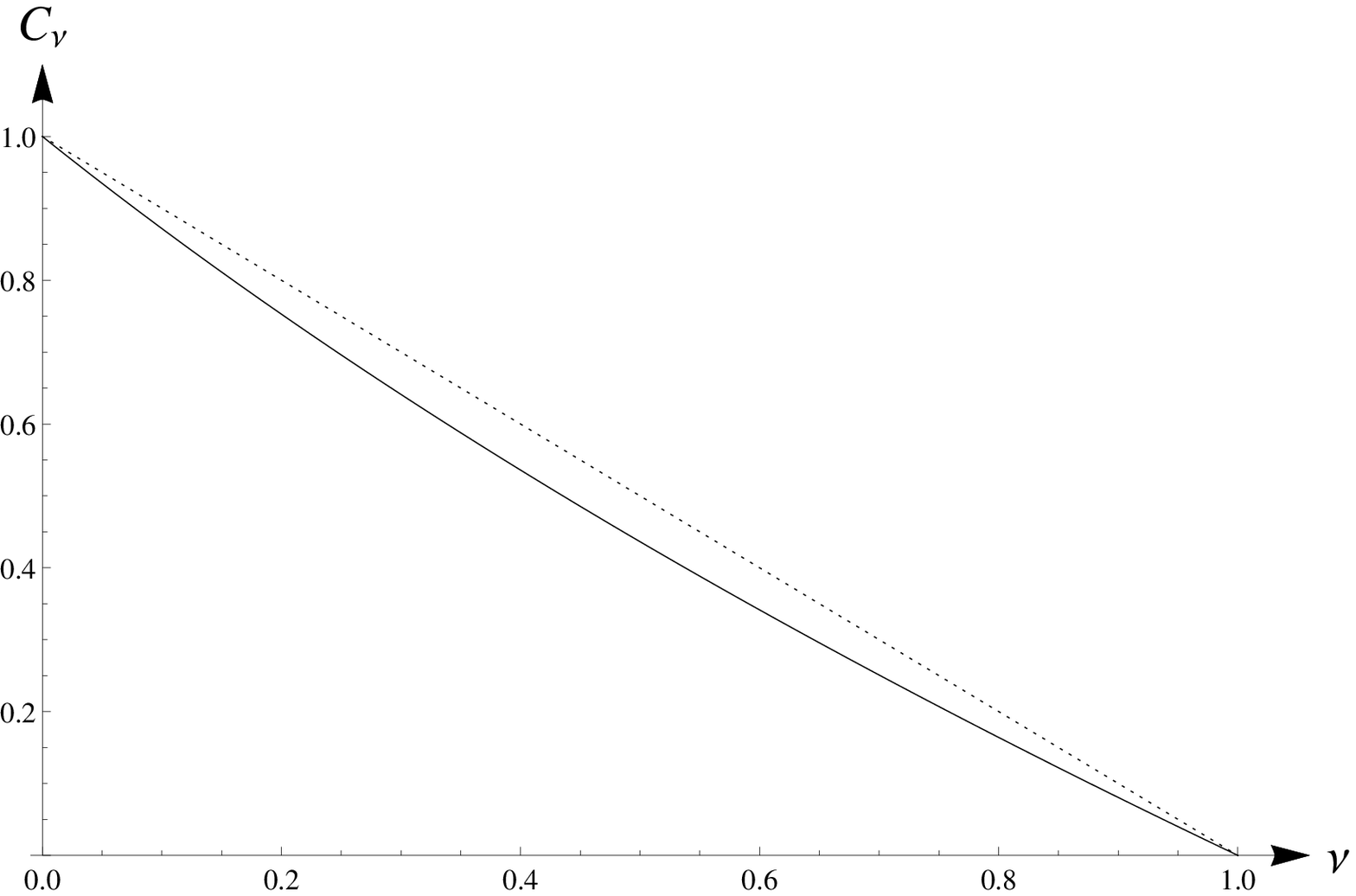}
\end{center}

We conclude the following estimates for massive Coulomb-Dirac operators:
\begin{cor}\label{c: massive estimates}
For every $\nu\in [0, 1)$ and $M >0$ the inequalities
\begin{align}
 |D^{\nu, M}| &\geqslant \Upsilon_\nu C_\nu\sqrt{-\Delta}\label{massive modulus}
\intertext{and}
 |D^{\nu, M}| &\geqslant \max\Big\{\frac{\Upsilon_\nu C_\nu}{1 +C_\nu}, 1 -2\nu\Big\}\sqrt{-\Delta +M^2}\label{massive bound 2}
\end{align}
hold with $C_\nu$ as in \eqref{c nu}.
\end{cor}

Note that $|D^{0, M}| =\sqrt{-\Delta +M^2}$ holds. The results analogous to Theorem \ref{absolute value bounds theorem} and Corollary \ref{c: massive estimates} were already discussed in the literature (see e.g. Lemma 1 in \cite{Bach1999} and Lemma 1 in \cite{brummelhuis2001stability}). However, they were obtained by taking the square root of the inequality
\begin{align*}
 (D^{\nu, M})^2 \geqslant \widetilde C_{\nu, M}(-\Delta),
\end{align*}
which cannot be valid with $\widetilde C_{\nu, M} >0$ for $\nu > \sqrt3/2$, since for such $\nu$ the domain of $D^\nu$ is not contained in $\mathsf H^1(\mathbb R^3, \mathbb C^4)$, see e.g. Lemma \ref{Lemma_operator_core} below. Even for $\nu\in [0, \sqrt3/2]$ the values of the constants in Theorem \ref{absolute value bounds theorem} and Corollary \ref{c: massive estimates} are improvements upon the above results.

Theorem \ref{absolute value bounds theorem} cannot hold for $\nu =1$, which is a consequence of the fact that the domain of $D^1$ is not contained in $\mathsf H^{1/2}(\mathbb R^3, \mathbb C^4)$, see Lemma \ref{Lemma_operator_core}. However, we can compare $|D^{1}|$ with the powers of the Laplacian smaller than $1/2$.
The comparison is based on the fact that for any $\lambda\in [0, 1)$ there exists $L_\lambda >0$ such that for all $a >0$ the inequality
\begin{align}\label{SSS lower bound}
 (-\Delta)^{1/2} -\frac2{\pi}|\cdot|^{-1} \geqslant L_\lambda a^{\lambda -1}(-\Delta)^{\lambda/2} -a^{-1}
\end{align}
holds, see Inequality (1.8) of \cite{FrankHLT} and Theorem 2.3 in \cite{SolovejSoerensenSpitzer}.

\begin{thm}\label{t: l bound}
 For any $\lambda\in [0, 1)$ and $a >0$ the inequality 
\begin{align*}
 |D^{1}| \geqslant K_\lambda a^{\lambda -1}(-\Delta)^{\lambda/2} -a^{-1}
\end{align*}
holds with 
\begin{align*}
                K_{\lambda}:=\min\Bigg\{L_\lambda\bigg(\frac{\pi(4-\sqrt{13})}{3(4-\pi)}\bigg)^{\lambda},\lambda^{-\lambda}
                (1-\lambda)^{-(1-\lambda)}\bigg(\frac{\sqrt{229}-8}{15}\bigg)^{\lambda}\Bigg\},
\end{align*}
where $L_\lambda$ is as in \eqref{SSS lower bound}.
\end{thm}

\subsection{Stability of the electron-positron field}

Let
\begin{align*}
 \mathbb H:=\int \mathrm dx:\Psi^*(x)D^{\nu, M}\Psi(x): +\frac\alpha2\int\mathrm dx\int\mathrm dy\frac{:\Psi^*(x)\Psi^*(y)\Psi(y)\Psi(x):}{|\mathbf x -\mathbf y|}
\end{align*}
be the Hamiltonian of the electron-positron field in the Furry picture, see \cite{brummelhuis2001stability, Bach1999}. Here $\Psi$ are the field operators and $\alpha >0$ is the fine-structure constant.
As explained in Section 2 of \cite{brummelhuis2001stability} (Note the misprint: the inequality in Section 2(ii) of \cite{brummelhuis2001stability} should coincide with inequality (1) there!), estimate \eqref{massive modulus} immediately implies the following result.
\begin{thm}\label{t: HF states}
Let $\nu\in [0, 1)$, $M, \alpha \geqslant0$. Then for
\begin{align}\label{alpha condition}
 \alpha\leqslant 4\Upsilon_\nu C_\nu/\pi
\end{align}
the energy $\mathcal E(\rho) :=\rho(\mathbb H)$ is non-negative for all generalised Hartree-Fock states $\rho$.
\end{thm}

For the physical value of the fine-structure constant $\alpha \approx 1/137$ condition \eqref{alpha condition} is fulfilled for $\nu \lesssim 0.97$, i.e. for the atomic numbers $Z =\nu/\alpha$ up to and including $132$.
In Theorem 1 of \cite{brummelhuis2001stability} (which is an improvement upon Theorem 2 in \cite{Bach1999}) the same result as in Theorem \ref{t: HF states} was proved under the stronger assumption
\begin{align*}
 \alpha \leqslant (4/\pi)(1 -\nu^2)^{1/2}(\sqrt{4\nu^2 +9} -4\nu)/3
\end{align*}
instead of \eqref{alpha condition}, which for $\alpha \approx 1/137$ corresponds to the atomic numbers $Z =\nu/\alpha \leqslant 117$ and cannot hold for any $\alpha >0$ if $\nu >\sqrt3/2$.

\subsection{Eigenvalue estimates in the Furry picture}

Theorems \ref{absolute value bounds theorem} and \ref{t: l bound} allow us to estimate the negative spectrum of perturbed projected massless Coulomb-Dirac operators in the Furry picture (cf. \cite{MorozovMueller}). Namely, for $\nu\in[0, 1]$ let $P^\nu_+$ be the spectral projector of $D^\nu$ to the half-line $[0, \infty)$.
Assuming that the negative energy states of $D^\nu$ are occupied (Dirac sea), we restrict the operator to the Hilbert space $\mathfrak H_+^\nu:= P^\nu_+\mathsf L^2(\mathbb R^3,\mathbb C^4)$. We now consider perturbations of such restricted operator by electromagnetic potentials. The following results are three-dimensional equivalents of Corollary 2 and Theorems 3 and 4 in \cite{MorozovMueller}. For numbers and self-adjoint operators we use the notation $x_\pm := \max\{\pm x, 0\}$ for the positive and negative parts of $x$.

\begin{lem}\label{l: defining corollary}
Suppose that $(\nu, \gamma)\in \big([0, 1]\times [0, \infty)\big)\setminus \big\{(1, 0)\big\}$.
Let $V$ be a measurable Hermitian $(4\times 4)$-matrix function with $\tr (V_+^{3 +\gamma})\in \mathsf L^1(\mathbb R^3)$ and such that there exists $C> 0$ with
\begin{align*}
 \int_{\mathbb R^3}\big\langle\upsilon(\mathbf x), V_-(\mathbf x)\upsilon(\mathbf x)\big\rangle\mathrm d\mathbf x  \leqslant C\Big(\big\||D^\nu|^{1/2}\upsilon\big\|^2 +\|\upsilon\|^2\Big), \textrm{ for all }\upsilon \in P_+^\nu\mathfrak D\big(|D^\nu|^{1/2}\big).
\end{align*}
Then the quadratic form
\begin{align*}
    \mathfrak d^\nu(V)&:P_+^\nu\mathfrak D\big(|D^\nu|^{1/2}\big)\to \mathbb R, \\ \mathfrak d^\nu(V)[\upsilon] &:=\big\||D^\nu|^{1/2}\upsilon\big\|^2 -\int_{\mathbb{R}^3}\big\langle \upsilon(\mathbf x), V(\mathbf x)\upsilon (\mathbf x)\big\rangle \mathrm{d}\mathbf{x}
\end{align*}
is closed and bounded from below in $\mathfrak H_+^\nu$.
\end{lem}
According to Theorem 10.1.2 in \cite{BirmanSolomjak}, there exists a unique self-adjoint operator $D^\nu(V)$ in $\mathfrak H_+^\nu$ associated to $\mathfrak d^\nu(V)$.

In the following theorem we estimate the negative spectrum of $D^\nu(V)$ by combining Theorems \ref{absolute value bounds theorem} and \ref{t: l bound} with the estimates
\begin{align*}
 \rank\big((-\Delta)^{t} -V\big)_- \leqslant C_t^\Delta\int_{\mathbb R^3}\tr \big(V_+(\mathbf x)\big)^{3/(2t)}\mathrm d\mathbf x,
\end{align*}
which according to Example 3.3 of \cite{Frank2014CLR} hold for all $0 <t <3/2$ with (non-optimal)
\begin{align*}
 C_t^\Delta :=\frac1{4\pi^2t}\Big(\frac3{3 -2t}\Big)^{(3 -t)/t}.
\end{align*}

We obtain the following estimates on the negative eigenvalues of $D^\nu(V)$.
\begin{thm}\label{CLR theorem}
For $\nu, \lambda\in [0, 1)$ let $C_\nu$ and $K_\lambda$ be defined in Theorems \ref{absolute value bounds theorem} and \ref{t: l bound}, respectively.\\
\emph{(a) Cwickel-Lieb-Rozenblum inequalities:}  For $\nu\in [0, 1)$,
\begin{equation}\label{CLR}
 \rank \big(D^\nu(V)\big)_- \leqslant C_{1/2}^\Delta C_\nu^{-3}\int_{\mathbb R^3} \tr \big(V_+(\mathbf x)\big)^3\dd\mathbf x.
\end{equation}
 \emph{(b) (Hardy-)Lieb-Thirring inequalities:} For $\nu\in [0, 1]$ and $\gamma> 0$, the estimate
\begin{equation*}
 \tr\big(D^\nu(V)\big)^\gamma_- \leqslant C^{\mathrm{LT}}_{\nu, \gamma}\int_{\mathbb R^3} \tr \big(V_+(\mathbf x)\big)^{3+ \gamma}\dd\mathbf x
\end{equation*}
holds with
\begin{align*}
 C^{\mathrm{LT}}_{\nu, \gamma} &:=\frac{6C_{1/2}^\Delta C_\nu^{-3}}{(\gamma +1)(\gamma +2)(\gamma +3)}, \quad \text{for }\nu< 1;\\ C^{\mathrm{LT}}_{1, \gamma} &:=\min_{\lambda\in (3/(3 +\gamma), 1)}\frac{\gamma^{1 +\gamma}\lambda^\gamma\,\Gamma(1 +3/\lambda)\,\Gamma(3 +\gamma -3/\lambda)\,C_{\lambda/2}^\Delta K_\lambda^{-3/\lambda}}{\Gamma(4 +\gamma)\,(3 -3\lambda)^{3(1- \lambda)/\lambda}\big((3 +\gamma)\lambda -3\big)^{3 +\gamma -3/\lambda}}.
\end{align*}
\end{thm}

It turns out that \eqref{CLR} cannot hold for $\nu =1$ for any constant in front of the integral, since $D^\nu(0)$ has a virtual level at zero. This is the result of the theorem below.

\begin{thm}\label{t: virtual level}
Suppose that $V =\begin{pmatrix}
                  V_{\mathrm{I,I}} & V_{\mathrm{I,II}}\\ V_{\mathrm{II,I}}& V_{\mathrm{II,II}}
                 \end{pmatrix}$
is a block matrix function on $\mathbb R^3$ with $2\times2$ blocks which satisfies the assumptions of Lemma \ref{l: defining corollary}. Using the spherical coordinates in $\mathbb R^3$ we define the $2\times2$ matrix functions
\begin{align*}
A_{(-1, -1)}(\rho, \theta, \phi)&:= \begin{pmatrix}
                                                                          \Big\langle\binom01, V_{\mathrm{I,I}}\binom01\Big\rangle & \Big\langle\binom0{\mathrm i}, V_{\mathrm{I,II}}\binom{\mathrm e^{-\mathrm i\phi}\sin\theta}{-\cos\theta}\Big\rangle\\ \Big\langle\binom{\mathrm e^{-\mathrm i\phi}\sin\theta}{-\cos\theta}, V_{\mathrm{II,I}}\binom0{\mathrm i}\Big\rangle & \Big\langle\binom{\mathrm e^{-\mathrm i\phi}\sin\theta}{-\cos\theta}, V_{\mathrm{II,II}}\binom{\mathrm e^{-\mathrm i\phi}\sin\theta}{-\cos\theta}\Big\rangle
                                                                         \end{pmatrix},\\
A_{(1, -1)}(\rho, \theta, \phi)&:= \begin{pmatrix}
                                                                          \Big\langle\binom10, V_{\mathrm{I,I}}\binom10\Big\rangle & \Big\langle\binom{\mathrm i}0, V_{\mathrm{I,II}}\binom{\cos\theta}{\mathrm e^{\mathrm i\phi}\sin\theta}\Big\rangle\\ \Big\langle\binom{\cos\theta}{\mathrm e^{\mathrm i\phi}\sin\theta}, V_{\mathrm{II,I}}\binom{\mathrm i}0\Big\rangle & \Big\langle\binom{\cos\theta}{\mathrm e^{\mathrm i\phi}\sin\theta}, V_{\mathrm{II,II}}\binom{\cos\theta}{\mathrm e^{\mathrm i\phi}\sin\theta}\Big\rangle
                                                                         \end{pmatrix},\\
A_{(-1, 1)}(\rho, \theta, \phi)&:= \begin{pmatrix}
                                                                          \Big\langle\binom{\mathrm e^{-\mathrm i\phi}\sin\theta}{-\cos\theta}, V_{\mathrm{I,I}}\binom{\mathrm e^{-\mathrm i\phi}\sin\theta}{-\cos\theta}\Big\rangle & \Big\langle\binom{\mathrm e^{-\mathrm i\phi}\sin\theta}{-\cos\theta}, V_{\mathrm{I,II}}\binom0{-\mathrm i}\Big\rangle\\ \Big\langle\binom0{-\mathrm i}, V_{\mathrm{II,I}}\binom{\mathrm e^{-\mathrm i\phi}\sin\theta}{-\cos\theta}\Big\rangle & \Big\langle\binom01, V_{\mathrm{II,II}}\binom01\Big\rangle
                                                                         \end{pmatrix},\\
A_{(1, 1)}(\rho, \theta, \phi)&:= \begin{pmatrix}
                                                                          \Big\langle\binom{\cos\theta}{\mathrm e^{\mathrm i\phi}\sin\theta}, V_{\mathrm{I,I}}\binom{\cos\theta}{\mathrm e^{\mathrm i\phi}\sin\theta}\Big\rangle & \Big\langle\binom{\cos\theta}{\mathrm e^{\mathrm i\phi}\sin\theta}, V_{\mathrm{I,II}}\binom{-\mathrm i}0\Big\rangle\\ \Big\langle\binom{-\mathrm i}0, V_{\mathrm{II,I}}\binom{\cos\theta}{\mathrm e^{\mathrm i\phi}\sin\theta}\Big\rangle & \Big\langle\binom10, V_{\mathrm{II,II}}\binom10\Big\rangle
                                                                         \end{pmatrix},
\end{align*}
where $V_{p,q}$ with $p,q\in\{\mathrm{I},\mathrm{II}\}$ are taken at $(\rho, \theta, \phi)$.
For $\mathbf j \in \{-1, 1\}^{2}$ and $\rho >0$ let
\begin{align*}
 W_{\mathbf j}(\rho) :=\int_0^\pi\int_0^{2\pi}A_{\mathbf j}(\rho, \theta, \phi)\mathrm d\phi\,\sin\theta\,\mathrm d\theta.
\end{align*}
If for some $\mathbf j \in \{-1, 1\}^{2}$ we have
\begin{align}\begin{split}\label{V integrability}
 &\|W_{\mathbf j}\|_{\mathbb C^{2\times2}}\in \mathsf L^1(\mathbb R_+, \rho^{2}\mathrm d\rho), \quad\int_0^\infty \Big\langle \binom1{j_2}, \big|W_{\mathbf j}(\rho)\big|\binom1{j_2}\Big\rangle\mathrm d\rho <\infty
\end{split}\end{align}
and
\begin{align}\label{VL condition}
\int_0^\infty\Big\langle\binom1{j_2}, W_{\mathbf j}(\rho)\binom1{j_2}\Big\rangle\mathrm d\rho >0,
\end{align}
then the operator $D^1(V)$ has at least one negative eigenvalue.
\end{thm}

In particular, $D^1(v\mathbb I_4)$ has negative eigenvalues for any non-zero continuous scalar function $v \geqslant 0$, which quickly decays at infinity.

The article is organised as follows: We do not provide the details of the proofs which are fully analogous to those of \cite{MorozovMueller}. This, in particular, applies to Lemma \ref{l: defining corollary} and Theorem \ref{CLR theorem}. In Section \ref{transforms section} we prepare useful representations of operators of interest with the help of certain unitary transforms. Section \ref{scalar section} is dedicated to the study of the operator $(-\Delta)^{1/2}-\alpha|\cdot|^{-1}$ in the representation, in which it can be relatively easily compared with $|D^\nu|$. Such comparison is done separately in different channels of the angular momentum decomposition in Section \ref{critical channels section}. Some less interesting technical parts of the proofs are relegated to the appendices. Finally, in Section \ref{main proofs section} we complete the proofs of Theorem \ref{absolute value bounds theorem}, Corollary \ref{c: massive estimates} and Theorems \ref{t: l bound} and \ref{t: virtual level}.

\paragraph{Acknowledgement:} S. M. was supported by the RSF grant 15-11-30007.

\section{Mellin, Fourier and related transforms in spherical coordinates}\label{transforms section}

Let $(\rho, \theta,\phi),(\varrho, \vartheta,\varphi)\in [0, \infty)\times[0,\pi)\times[0, 2\pi)$ be the spherical coordinates for $\mathbb R^3$ in coordinate and momentum spaces, respectively.

\paragraph{Fourier transform.}
We use the standard unitary Fourier transform in $\mathsf L^2(\mathbb R^3)$ given in the spherical coordinates for $\upsilon\in \mathsf L^1(\mathbb R^3)\cap \mathsf L^2(\mathbb R^3)$ by
\begin{align}\label{Fourier transform}
\begin{split}
 (\mathcal F\upsilon)(\varrho, \vartheta,\varphi):=& \frac1{(2\pi)^{3/2}}\int_0^\infty\int_0^{\pi}\int_0^{2\pi}\mathrm e^{-\mathrm i \varrho \rho(\sin\theta\sin\vartheta\cos(\phi- \varphi)+\cos\theta\cos\vartheta)}\\
 & \times\upsilon(\rho, \theta,\phi)\, \mathrm{d}\phi\, \sin\theta\,\mathrm{d}\theta\, \rho^2\mathrm{d}\rho.
\end{split}
\end{align}

Let $\mathbb S^2$ be the unit sphere in $\mathbb R^3$.
As an orthonormal basis in $\mathsf{L}^2(\mathbb{S}^{2};\mathbb{C}^{2})$ we use the spherical spinors $\Omega_{l,m,s}$, which are defined by (2.1.25) and (2.1.26) in \cite{BalinskyEvans}, with $l\in\mathbb{N}_{0}$, $\ m\in \{-l-1/2,\ldots,l+1/2\}$ and $s\in\{-1/2,1/2\}$. The corresponding  index set is denoted by 
\begin{align*}
\mathfrak{T}:=
                  \bigg\{(l,m,s): l\in \mathbb{N}_{0}, m\in \Big\{-l-\frac{1}{2},\ldots,l+
                  \frac{1}{2}\Big\},s=\pm\frac{1}{2},\Omega_{l,m,s}\neq 0 \bigg\},           
\end{align*}
see (2.1.27) in \cite{BalinskyEvans}.
\begin{lem}\label{Fourier channel lemma}
 For $(l,m,s)\in \mathfrak{T}$ and $\psi\in \mathsf C_0^\infty\big([0, \infty)\big)$ the Fourier transform of
\begin{equation}\label{auxiliary nonsense}
 \Psi_{l,m,s}(\rho, \theta,\phi): = \rho^{-1}\psi(\rho)\Omega_{l,m,s}(\theta,\phi)
\end{equation}
 is given in the spherical coordinates by
\begin{equation}\label{Fourier channel}
 (\mathcal F\Psi_{l,m,s})(\varrho, \vartheta,\varphi)= (-\mathrm{i})^{l}
 \int_0^\infty \sqrt{\frac{\rho}{\varrho}}J_{l+1/2}(\rho \varrho)\psi(\rho)\mathrm d\rho\mathrm \, \Omega_{l,m,s}(\vartheta,\varphi).
\end{equation}
Here $J_{l+1/2}$ is the Bessel function of the first kind. 
\end{lem}

\begin{proof}
 According to 10.60.7, 14.7.1, 14.7.17, 14.30.9 and 10.47.3 in \cite{dlmf} the relation 
\begin{align}\label{aux1}
\begin{split}
&\exp\Big(-\mathrm i \varrho \rho\big(\sin\theta\sin\vartheta\cos(\phi- \varphi)+\cos\theta\cos\vartheta\big)\Big)\\&=(2\pi)^{3/2} \sum\limits_{l=0}^{\infty}
(-\mathrm{i})^{l} (\rho \varrho  )^{-1/2}J_{l+1/2}(\rho\varrho)
 \sum\limits_{m=-l}^{l} \overline{Y_{l,m}(\theta,\phi)}Y_{l,m}(\vartheta,\varphi)
 \end{split}
\end{align}
holds. Here for $l\in \mathbb{N}_{0}$ and $m\in \{-l,-l+1,\dots,l-1,l\}$ the functions $Y_{l,m}$ are the spherical harmonics, see 14.30.1 in \cite{dlmf}.
Substituting \eqref{auxiliary nonsense} into \eqref{Fourier transform} and using \eqref{aux1} together with (2.1.25), (2.1.26) in \cite{BalinskyEvans} and 14.30.8 in \cite{dlmf} we obtain \eqref{Fourier channel}.
\end{proof}

\paragraph{Mellin transform.} Let $\mathcal M$ be the unitary Mellin transform, first defined on $\mathsf C_0^\infty(\mathbb R_+)$ by
\begin{equation}\label{Mellin transform}
 (\mathcal M\psi)(\tau):= \frac1{\sqrt{2\pi}}\int_0^\infty r^{-1/2- \mathrm i\tau}\psi(r)\mathrm dr,
\end{equation}
and then extended to a unitary operator $\mathcal M:\mathsf L^2(\mathbb R_+)\to \mathsf L^2(\mathbb R)$, see e.g. \cite{YaouancOliverRaynal}.

\begin{definition}\label{D lambda definition}
For $\lambda\in \mathbb R\setminus\{0\}$ let $\mathfrak D^\lambda$ be the set of functions $\psi\in \mathsf L^2(\mathbb R)$ such that there exists $\Psi$ analytic in the strip $\mathfrak S^\lambda:= \big\{z\in\mathbb C: \Im z/\lambda\in (0, 1)\big\}$ with the properties
\begin{enumerate}
 \item $\underset{t\to +0}\Ltwolim\ \Psi(\cdot+ \mathrm it\lambda)= \psi$;
 \item there exists $\underset{t\to 1-0}{\Ltwolim}\ \Psi(\cdot+ \mathrm it\lambda)$;
 \item $\sup\limits_{t\in (0, 1)}\displaystyle\int_{\mathbb R}\big|\Psi(\tau+ \mathrm it\lambda)\big|^2\mathrm d\tau< \infty$.
\end{enumerate}
\end{definition}

For $\lambda\in \mathbb R$ let the operator of multiplication by $r^\lambda$ in $\mathsf L^2(\mathbb R_+, \textrm dr)$ be defined on its maximal domain $\mathsf L^2\big(\mathbb R_+, (1+ r^{2\lambda})\mathrm dr\big)$.
Let $R^\lambda: \mathfrak D^\lambda\to \mathsf L^2(\mathbb R)$ be the linear operator defined by
\begin{equation*}
 R^\lambda\psi:= \begin{cases}\underset{t\to 1-0}{\Ltwolim}\ \Psi(\cdot+ \mathrm it\lambda), &\lambda \neq 0;\\ \psi, &\lambda =0,\end{cases}
\end{equation*}
with $\Psi$ as in Definition \ref{D lambda definition}.
According to (14) in \cite{MorozovMueller},
\begin{equation}\label{complex shifted}
 \mathcal Mr^\lambda\mathcal M^*= R^\lambda
\end{equation}
holds for all $\lambda\in\mathbb R$ (see also \cite{YaouancOliverRaynal}, Section II). 

The following lemma can be proved in the same way as Lemma 8 in \cite{MorozovMueller}.
\begin{lem}\label{Mellin-Bessel lemma}
Let $l\in\mathbb{N}_{0}$. The relation
\begin{equation*}
 \bigg(\mathcal M\Big((-\mathrm i)^l\int_0^\infty\sqrt{\cdot r}J_{l+1/2}(\cdot r)\psi(r)\mathrm dr\Big)\bigg)(\tau) = \Xi_l(\tau)(\mathcal M\psi)(-\tau)
\end{equation*}
holds for every $\psi\in C_0^\infty\big([0, \infty)\big)$ and $\tau\in\mathbb R$ with
\begin{equation}\label{Xi}
 \Xi_l(\tau):= (-\mathrm i)^{l}2^{-\mathrm i\tau}\dfrac{\Gamma\big((l+ 3/2- \mathrm i\tau)/2\big)}{\Gamma\big((l+ 3/2+ \mathrm i\tau)/2\big)}.
\end{equation}
\end{lem}
\begin{rem}\label{Xi bar and inverse remark}
For any $l \in\mathbb N_0$ the function $\Xi_l$ introduced in \eqref{Xi} allows a unique analytic continuation to $\mathbb C\setminus \big(-\mathrm i(l+3/2 +2\mathbb N_0)\big)$, whereas
\begin{align*}
 \Xi_l^{-1} =\overline{\Xi_l(\overline\cdot)}
\end{align*}
allows a unique analytic continuation to $\mathbb C\setminus \big(\mathrm i(l +3/2 +2\mathbb N_0)\big)$.
\end{rem}
Note that for all $l\in\mathbb{N}_0$ the function $\Xi_l^{-1}$ has no pole in $\overline{\mathfrak{S}^1}$. This situation is different (and simpler) than the one investigated in \cite{MorozovMueller}. The reason can be traced back to the non-existence of the Hardy inequality in $\mathbb{R}^2$ in opposite to $\mathbb{R}^3$. As a consequence, $\psi\in\mathfrak D^1$ always implies $\Xi_l^{-1}\psi\in \mathfrak D^1$. 
Thus the assumptions in the following lemma can be relaxed in comparison to Lemma 10 and Corollary 11 in \cite{MorozovMueller}. The arguments used in the proofs of these statements can also be successfully applied here.
\begin{lem}\label{V corollary}
For $l\in \mathbb N_0$ and $\psi\in\mathfrak D^1$ the identity
\begin{align*}
 \Xi_lR^1\Xi_l^{-1}\psi =V_{l}(\cdot +\mathrm i/2)R^1\psi
\end{align*}
holds with
\begin{equation}\label{V}
 V_l(z):=\frac{\Gamma\big((l+1+ \mathrm iz)/2\big)\Gamma\big((l+1- \mathrm iz)/2\big)}{2\Gamma\big((l+2+ \mathrm iz)/2\big)\Gamma\big((l+2- \mathrm iz)/2\big)},
\end{equation}
for $z\in (\mathbb C\setminus \mathrm i\mathbb Z)\cup\{0\}$.
\end{lem}
The following lemma can be proved analogously to Lemma 12 in \cite{MorozovMueller}.
\begin{lem}\label{V monotonicity lemma}
For every $l\in \mathbb N_0$ the function \eqref{V} is analytic in $(\mathbb C\setminus \mathrm i\mathbb Z)\cup\{0\}$ and has the following properties:
\begin{enumerate}
 \item $V_l(z) =V_l(-z)$, for all $z\in (\mathbb C\setminus \mathrm i\mathbb Z)\cup\{0\}$;
 \item $V_l(\tau)$ is positive and strictly monotonously decreasing for $\tau\in \mathbb R_+$;
 \item $V_l(\mathrm i\zeta)$ is positive and strictly monotonously increasing for $\zeta\in [0, 1)$;
 \item The relation
\begin{equation}\label{V via 1/V}
 (z^2 +(l+1)^2)V_{l}(z) = \big(V_{l +1}(z)\big)^{-1}
\end{equation}
holds for all $z\in (\mathbb C\setminus \mathrm i\mathbb Z)\cup\{0\}$.
\end{enumerate}
\end{lem}
\paragraph{Angular decomposition.}

We can represent arbitrary $u\in\mathsf{L}^2(\mathbb{R}^3,\mathbb{C}^2)$ in the spherical coordinates as 
\begin{align}\label{spherical series}
         u(\rho, \theta,\phi)&=\sum\limits_{(l,m,s)\in \mathfrak{T}}\rho^{-1} u_{l,m,s}(\rho) \Omega_{l,m,s}(\theta,\phi);
\end{align}
with 
\begin{align*}
         u_{l,m,s}(\rho)&:=\rho\int\limits_{0}^{2\pi}\int\limits_{0}^{\pi}\big\langle \Omega_{l,m,s}
           (\theta,\phi),u(\rho, \theta,\phi)\big\rangle_{\mathbb C^2} 
           \sin(\theta)\,\mathrm{d}\theta\,\mathrm{d}\phi. 
\end{align*}

The angular momentum decomposition associated to \eqref{Dirac symbol} is given by
\begin{align}\label{A}
         \mathcal A: \mathsf{L}^2(\mathbb{R}^3,\mathbb{C}^4) \rightarrow \bigoplus\limits_{(l,m,s)\in\mathfrak{T}}
         \mathsf{L}^2(\mathbb{R}_{+},\mathbb{C}^2);\quad \begin{pmatrix}
         \psi_1 \\ \psi_2 \\ \psi_3 \\ \psi_4
         \end{pmatrix} \mapsto 
         \bigoplus_{(l,m,s)\in\mathfrak{T}}
         \begin{pmatrix}
                \begin{pmatrix}
                        \psi_1 \\ \psi_2 
                \end{pmatrix}_{l,m,s} \\ 
                 \mathrm{i}\begin{pmatrix}
                        \psi_3\\ \psi_4
                \end{pmatrix}_{l+2s,m,-s}
         \end{pmatrix}.
\end{align}
Let $\nu\in [0, 1]$. For $l\in\mathbb{N}_{0}$ and $s\in\{-1/2,1/2\}$ we define the operators $\widetilde{D}_{l,s}^\nu$ in $\mathsf L^2(\mathbb R_+, \mathbb C^2)$ by the differential expressions
\begin{equation}\label{d_kappa}
 \begin{pmatrix}
                   -\dfrac\nu r& -\dfrac{\mathrm d}{\mathrm dr}- \dfrac{2sl+s+1/2}{r}\\ \dfrac{\mathrm d}{\mathrm dr}- \dfrac{2sl+s+1/2}{r} & -\dfrac\nu r
                  \end{pmatrix}
\end{equation}
on $\mathsf{C}_{0}^{\infty}(\mathbb{R}_{+},\mathbb{C}^2)$. Furthermore, we introduce 
$\widetilde{D}^\nu$ as the operator corresponding to \eqref{Dirac symbol} on the domain $\mathsf C_0^\infty\big(\mathbb R^3\setminus\{0\},\mathbb{C}^4\big)$.

The next lemma follows from Section 2.1 in \cite{BalinskyEvans}.
\begin{lem}\label{angular Dirac lemma}
For $\nu\in[0,1]$ the operator $(\widetilde{D}^\nu)^{*}$ satisfies
\begin{equation*}
 \mathcal A\,(\widetilde{D}^\nu)^{*}\,\mathcal A^*= \underset{(l,m,s)\in \mathfrak{T}}\bigoplus (\widetilde{D}_{l,s}^\nu)^{*}.
\end{equation*}
\end{lem}

\begin{rem}\label{r: D^nu}
By Theorem 4 in \cite{EstebanLoss} there exists a unique self-adjoint extension $D^{\nu, 1}$ of $\widetilde{D}^\nu +\beta$ which has a positive Schur complement for $\nu \in [0,1]$. Moreover, for $\nu \in [0,1)$ the operator $D^{\nu} :=D^{\nu, 1} -\beta$ is the only self-adjoint extension of $\widetilde{D}^\nu$ with the property that every function in its domain possesses finite kinetic energy, i.e. belongs to $\mathsf H^{1/2}(\mathbb R^3, \mathbb C^4)$ (see Corollary 6 in \cite{Mueller}).  
\end{rem}

The following property of $D^{\nu}$ follows from Theorem 5 in \cite{Mueller}.
\begin{lem}\label{Lemma_operator_core}
Let $\nu\in [0,1]$. The set
\begin{align*} 
             \mathfrak{C}^{\nu}:= \mathsf{C}_{0}^{\infty}\big(\mathbb{R}^3\setminus\{0\},\mathbb{C}^{4}\big)\dot{+}\begin{cases} \{0\},&\!\!\!\text{if }\nu\in[0,\sqrt{3}/2];\\ 
             \Span\big\{\Psi_{\mathbf k}^{\nu} : \mathbf k \in \{-1/2,1/2\}^{2}\big\},&\!\!\!\text{otherwise}
             \end{cases}
\end{align*} 
with
\begin{align}\label{special_function}
             \Psi_{\mathbf k}^{\nu}(\rho,\theta,\phi)&:= \sqrt{2\pi}
             \mathrm{e}^{-\rho} \rho^{\Upsilon_\nu -1}
             \begin{pmatrix}
             \nu \Omega_{\frac{1}{2}+k_2, k_1, -k_2}
             (\theta,\phi)\\
             -\mathrm{i}\big(\Upsilon_\nu +(-1)^{\frac{1}{2} -k_2}\big)  \Omega_{\frac{1}{2} -k_2,k_1,k_2}(\theta,\phi)
             \end{pmatrix}
\end{align} 
is an operator core for $D^{\nu}$.
\end{lem}

\begin{proof}
It is known that $\mathsf C_0^\infty\big(\mathbb R^3\setminus\{0\}, \mathbb C^4\big)$ is dense in $\mathsf H^1(\mathbb R^3, \mathbb C^4)$, combine e.g. Theorem 3.23 and Example 5.26 in \cite{AdamsFournier} with Proposition 9 of Appendix A in \cite{Rauch}.
By the Hardy inequality the graph norm of $\widetilde{D}^\nu$ is subordinate to the norm of $\mathsf H^{1}(\mathbb R^3, \mathbb C^4)$. Hence the closure of $\mathsf C_0^\infty\big(\mathbb R^3\setminus\{0\}, \mathbb C^4\big)$ with respect to the graph norm of $D^\nu$ contains $\mathsf H^{1}(\mathbb R^3, \mathbb C^4)$.

By Theorem 5 in \cite{Mueller}, \eqref{operator core} would be an operator core for $D^{\nu}$ if we replace $\mathrm{e}^{-\rho}$ in \eqref{special_function} by $\xi(\rho)$, where $\xi\in\mathsf C^\infty\big(\mathbb R_+, [0, 1]\big)$ satisfies $\xi(r) =1$ for $r\in (0, 1)$ and $\xi(r) =0$ for $r >2$.

Now it is enough to establish that for $\mathbf k \in \{-1/2,1/2\}^{2}$ the functions
\begin{align*}
 \Phi_{\mathbf k}^{\nu}(\rho,\theta,\phi)&:= 
             \big(\mathrm{e}^{-\rho} -\xi(\rho)\big) \rho^{\Upsilon_\nu -1}
             \begin{pmatrix}
             \nu \Omega_{\frac{1}{2}+k_2,k_1,-k_2}
             (\theta,\phi)\\
             -\mathrm{i}\big(\Upsilon_\nu +(-1)^{\frac{1}{2}-k_2}\big)  \Omega_{\frac{1}{2}-k_2,k_1,k_2}(\theta,\phi)
             \end{pmatrix}
\end{align*}
 belong to $\mathsf H^{1}(\mathbb R^3, \mathbb C^4)$. To do this we observe that by Lemma \ref{angular Dirac lemma}
\begin{align*}
 \|\nabla\Phi_{\mathbf k}^{\nu}\|^2 =\|D^0\Phi_{\mathbf k}^{\nu}\|^2 =\bigg\|D^0_{\frac{1}{2}+k_2,k_1,-k_2}\big(\mathrm{e}^{-\rho} -\xi(\rho)\big) \rho^{\Upsilon_\nu}\binom{\nu}{\Upsilon_\nu +(-1)^{\frac{1}{2} -k_2}}\bigg\|^2
\end{align*}
holds, where the finiteness of the right hand side follows by a straightforward calculation based on \eqref{d_kappa}.
\end{proof}

Let $(l,m,s)\in\mathfrak{T}$ and $\nu\in [0, 1]$. Lemma \ref{angular Dirac lemma} implies that there is a unique self-adjoint extension $D_{l,s}^\nu$ of $\widetilde{D}_{l,s}^\nu$ such 
that
\begin{align}\label{D nu}
            D_{l,s}^\nu= \big(\mathcal{A}D^\nu \mathcal A^*\big)_{l,m,s}
\end{align}
holds. By Lemma \ref{Lemma_operator_core} we conclude that 
\begin{equation}\label{operator core}
 \mathfrak C^\nu_{l,s}:= \mathsf C_0^\infty(\mathbb R_+, \mathbb C^2) \dot + \begin{cases}
 \Span\{\psi^\nu_s\}, &\text{if }\nu\in(\sqrt{3}/2,1]\text{ and }l =0, s=1/2\\
                      &\text{or }l=1,s=-1/2;\\ \{0\}, &\text{otherwise},
                                                                       \end{cases}
\end{equation}
with
\begin{align}\label{extra functions on the half line}
 \psi^\nu_{s}(\rho):=\sqrt{2\pi}\mathrm e^{-\rho} \rho^{\Upsilon_\nu}
 \binom{\nu}{\Upsilon_\nu +(-1)^{\frac{1}{2}+s}}\text{ for }\rho\in\mathbb R_+
\end{align}
is an operator core for $D_{l,s}^\nu$.

\paragraph{MWF-transform.}
We now introduce the unitary transform 
\begin{equation*}
\mathcal T:\mathsf L^2(\mathbb R^3,\mathbb{C}^4)\to \underset{(l,m,s)\in\mathfrak{T}}\bigoplus\mathsf L^2(\mathbb R,\mathbb{C}^2), \quad \mathcal T:= \mathcal M\mathcal A\diag(1,\, 1,\, -\mathrm i,\, -\mathrm i)\mathcal F,
\end{equation*}
where $\mathcal M$ acts fibre-wise.
A direct calculation using Lemmata \ref{Fourier channel lemma} and \ref{Mellin-Bessel lemma} gives
\begin{equation*}
\mathcal T\begin{pmatrix}\upsilon_1\\ \upsilon_2\\ \upsilon_3\\ \upsilon_4\end{pmatrix}= \underset{(l,m,s)\in\mathfrak{T}}\bigoplus\begin{pmatrix}
                                                                                                                                    \mathcal T_l\displaystyle\binom{\upsilon_1}{\upsilon_2}_{l,m,s}\\ \mathcal T_{l +2s}\displaystyle\binom{\upsilon_3}{\upsilon_4}_{l+2s,m,-s}
                                                                                                                                   \end{pmatrix}
,
\end{equation*}
where for $(l,m,s)\in\mathfrak{T}$ the operators $\mathcal T_l: \mathsf L^2(\mathbb R_+)\to \mathsf L^2(\mathbb R)$ are given by
\begin{equation}\label{T components}
 (\mathcal T_l\psi)(\tau):= \Xi_l(\tau)(\mathcal M\psi)(-\tau) \quad \textrm{for any }\psi\in \mathsf L^2(\mathbb R_+).
\end{equation}

In the following two lemmata we study the actions of several operators in the MWF-representation.
\begin{lem}\label{transformed operators lemma}
 The relations
\begin{equation*}
\mathcal T(-\mathrm i\boldsymbol\alpha\cdot\nabla)\mathcal T^*= \underset{
(l,m,s)\in\mathfrak{T}}\bigoplus (R^1\otimes\sigma_1)
\end{equation*}
and for any $\lambda\in \mathbb R$
\begin{equation*}
 \mathcal T\big((-\Delta)^{\lambda/2}\otimes \mathds{1}_{\mathbb{C}^4}\big)\mathcal T^*= \underset{(l,m,s)\in\mathfrak{T}}\bigoplus R^\lambda\otimes \mathds{1}_{\mathbb{C}^2}
\end{equation*}
hold.
\end{lem}
\noindent Lemma \ref{transformed operators lemma} follows immediately by Lemma 8 in \cite{Mueller} and \eqref{complex shifted}. The next lemma can be proved as Lemma 19 in \cite{MorozovMueller}.
\begin{lem}\label{1/r with T lemma}
The relation
\begin{equation*}
\mathcal T\big(|\cdot|^{-1}\otimes \mathds{1}_{\mathbb{C}^4}\big)\mathcal T^*= \underset{(l,m,s)\in\mathfrak{T}}\bigoplus \begin{pmatrix}
                                                                                                                           \Xi_lR^1\Xi_l^{-1} &0 \\ 0 &\Xi_{l +2s}R^1\Xi_{l +2s}^{-1}
                                                                                                                          \end{pmatrix}
\end{equation*}
holds.
\end{lem}

\paragraph{U-transform.}

Let $(l,m,s)\in \mathfrak{T}$. We define the following unitary operator 
\begin{align}\label{Upsi}
\mathcal U_{l,s}:\mathsf L^2(\mathbb R_+, \mathbb C^2)\to \mathsf L^2(\mathbb R, \mathbb C^2);\quad
 \mathcal U_{l,s}\binom{\psi_1}{\psi_2}:=\binom{\mathcal T_{l}\psi_1}{-\mathrm i\mathcal T_{l+2s}\psi_2}.
\end{align}
Note the relation 
\begin{align}\label{Upsi extra}
\mathcal U_{l,s}\binom{\psi_1}{\psi_2} =\Bigg(\mathcal{TA}^*\bigoplus_{(l',m',s')\in\mathfrak{T}}\delta_{l', l}\delta_{m', m}\delta_{s', s}\binom{\psi_1}{\psi_2}\Bigg)_{l,m,s}.
\end{align}

A straightforward calculation involving \eqref{T components}, \eqref{Mellin transform}, \eqref{V}, \eqref{V via 1/V} and the elementary properties of the gamma function delivers
\begin{lem}\label{T transformed extra functions lemma}
Let $\nu\in (\sqrt{3}/2, 1]$. For $s\in\{-1/2, 1/2\}$ the functions \eqref{extra functions on the half line} from the operator core $\mathfrak C^\nu_{1/2 -s, s}$ of $D^\nu_{1/2 -s, s}$ satisfy the relation
\begin{equation*}
 \mathcal U_{1/2 -s, s}\psi^\nu_{s} =\binom{\xi_s^\nu}{\eta_s^\nu} + \chi_s^\nu\binom{1}{\nu V_{1/2 -s}(\mathrm i\Upsilon_\nu)}
\end{equation*}
where we have introduced
\begin{align}
 \chi_s^\nu(\tau) := &\nu\Xi_{1/2 -s}\big(\mathrm i(\Upsilon_\nu +1/2)\big)\frac{(\tau -\mathrm i)\Gamma(\mathrm i\tau+ \Upsilon_\nu +1/2)}{\mathrm i(\Upsilon_\nu -1/2)},\label{chi}\\
 \xi_s^\nu(\tau) := &\nu\Gamma(\mathrm i\tau+ \Upsilon_\nu+ 1/2)\Big(\Xi_{1/2 -s}(\tau) -\frac{(\tau -\mathrm i)\Xi_{1/2 -s}\big(\mathrm i(\Upsilon_\nu +1/2)\big)}{\mathrm i(\Upsilon_\nu -1/2)}\Big),\label{xi}\\
 \eta_s^\nu(\tau) := &\mathrm i(2s -\Upsilon_\nu)\Gamma(\mathrm i\tau+ \Upsilon_\nu+ 1/2)\Big(\Xi_{s +1/2}(\tau) -\frac{(\tau -\mathrm i)\Xi_{s +1/2}\big(\mathrm i(\Upsilon_\nu +1/2)\big)}{\mathrm i(\Upsilon_\nu -1/2)}\Big).\label{eta}
\end{align}
\end{lem}

\section{On the Friedrichs extension of the relativistic Coulomb operator in the Fourier-Mellin space}\label{scalar section}

For $\alpha\in \mathbb R$ in $\mathsf L^2(\mathbb R^3, \mathbb C)$ consider the symmetric operator
\begin{equation*}
 \widetilde H^\alpha :=(-\Delta)^{1/2} - \alpha|\cdot|^{-1}
\end{equation*}
on the domain $\mathsf H^1(\mathbb R^3, \mathbb C)$.
Let $\mathfrak A:= \big\{(l, m)\in \mathbb Z^2: l \geqslant 0, -l \leqslant m\leqslant l\big\}$. For every scalar function $u\in \mathsf L^2(\mathbb R^3, \mathbb C)$ the decomposition in the spherical harmonics (see Section 14.30 in \cite{dlmf}) is given by
\begin{align}\label{scalar spherical}
 u(\rho, \theta, \phi) &=\sum_{(l, m)\in\mathfrak A}\rho^{-1}u_{l, m}(\rho)Y_{l,m}(\theta, \phi),\\
 u_{l, m}(\rho) &:=\rho \int_0^{2\pi}\int_0^\pi\overline{Y_{l, m}(\theta, \phi)} u(\rho, \theta, \phi)\sin\theta\,\mathrm d\theta\,\mathrm d\phi.
\end{align}
Introducing the corresponding unitary operator
\begin{align}
 \mathcal R: \mathsf L^2(\mathbb R^3, \mathbb C)\to \underset{(l, m)\in \mathfrak A}\bigoplus \mathsf L^2(\mathbb R_+, \mathbb C),\quad \mathcal R: u\mapsto \underset{(l, m)\in \mathfrak A}\bigoplus u_{l, m}
\end{align}
we observe the relation (cf. the beginning of Section 3 in \cite{MorozovMueller})
\begin{equation}\label{MWF Herbst}
 (\mathcal{MRF})\widetilde H^\alpha(\mathcal{MRF})^*= \underset{(l, m)\in \mathfrak A}\bigoplus \big( 1 -\alpha V_{l}(\cdot + \mathrm i/2)\big)R^1 =:\underset{(l, m)\in \mathfrak A}\bigoplus\widetilde H^\alpha_{l},
\end{equation}
where the right hand side is an orthogonal sum of operators in $\mathsf L^2(\mathbb R, \mathbb C)$ densely defined on $\mathfrak D^1$.

In the next lemma, which follows from Lemmata \ref{V corollary} and \ref{V monotonicity lemma} in the same way as Lemma 21 in \cite{MorozovMueller}, we state the optimal condition on $\alpha$ such that $\widetilde H^\alpha_l$ is bounded from below.
\begin{lem}\label{critical constants lemma}
 For $l\in \mathbb{N}_{0}$ and $\alpha \in\mathbb R$ the operator $\widetilde H^\alpha_l$ is symmetric. It is bounded below (and non-negative) in $\mathsf L^2(\mathbb R)$ if and only if
\begin{equation*}
 \alpha\leqslant\alpha_l:= \frac1{V_{l}(0)}= \frac{2\Gamma^2\big((l +2)/2\big)}{\Gamma^2\big((l+ 1)/2\big)}.
\end{equation*}
\end{lem}
Given $l\in \mathbb N_{0}$, Lemma \ref{critical constants lemma} allows us for $\alpha\leqslant\alpha_l$ to pass from the symmetric operator $\widetilde H^\alpha_l$ to the self-adjoint operator $H^\alpha_l$ by Friedrichs extension \cite{Friedrichs}. 

\begin{lem}\label{critical lower bound lemma}
 For $\lambda\in (0, 1)$ the inequality
 \begin{equation}\label{critical lower bound with operators}
  H_0^{\alpha_0}\geqslant L_\lambda a^{\lambda -1}R^\lambda -a^{-1}
 \end{equation}
 holds for all $a >0$ with $L_\lambda$ as in \eqref{SSS lower bound}.
\end{lem}

\begin{proof}
 For any $\upsilon\in \mathfrak D^{1}$ we have
\begin{equation}\label{form of H_m tilde}
 \langle \upsilon, \widetilde H_0^{\alpha_0}\upsilon\rangle = \langle \upsilon, R^1\upsilon\rangle - \alpha_0\langle \upsilon, V_{0}(\cdot +\mathrm i/2)R^1\upsilon\rangle.
\end{equation}
Letting
\begin{equation*}
 \Phi:= (\mathcal{MRF})^*\bigoplus_{(l,m)\in\mathfrak{A}}\delta_{l, 0}\delta_{m, 1/2}\upsilon \in\mathsf H^1(\mathbb R^3, \mathbb C),
\end{equation*}
by scalar analogues of Lemmata \ref{transformed operators lemma} and \ref{1/r with T lemma} together with Lemma \ref{V corollary} we obtain
\begin{align*}
\langle \upsilon, R^1\upsilon\rangle&=\langle \Phi, (-\Delta)^{1/2}\Phi\rangle; \\
 \langle \upsilon, V_{0}(\cdot +\mathrm i/2)R^1\upsilon\rangle &=\langle \Phi, r^{-1}\Phi\rangle.
\end{align*}
Thus \eqref{SSS lower bound} implies that \eqref{form of H_m tilde} can be estimated from below by
\begin{equation*}
 \langle \Phi, \big(L_\lambda a^{\lambda -1}(-\Delta)^{\lambda/2}-a^{-1}\big)\Phi \rangle.
\end{equation*}
Using the scalar version of Lemma \ref{transformed operators lemma} again and that $H_0^{\alpha_0}$ is the Friedrichs extension of $\widetilde H_0^{\alpha_0}$ we conclude \eqref{critical lower bound with operators}.
\end{proof}

The following description of the domains of $H^\alpha_l$ with $\alpha \leqslant \alpha_l$ follows analogously to Theorem 2 and Corollary 2 in \cite{YaouancOliverRaynal} (see also Section 2.2.3 of \cite{BalinskyEvans}), where the case $l =0$ is studied.

\begin{lem}\label{H alpha l domains lemma}
 Let $l\in\mathbb N_{0}$. 
 \begin{enumerate}
\item[(i)] For $\alpha < V^{-1}_{l}(\mathrm{i}/2)$ the operator $\widetilde H^\alpha_l$ is self-adjoint, i.e. $\widetilde H^\alpha_l=H^\alpha_l$.
\item[(ii)] The operator $\widetilde H^{V^{-1}_{l}(\mathrm{i}/2)}_l$ is essentially self-adjoint, i.e. the relation \\ $\big(\widetilde H^{V^{-1}_{l}(\mathrm{i}/2)}_l\big)^{*}=H^{V^{-1}_{l}(\mathrm{i}/2)}_l$ holds.
\item[(iii)] For $\alpha \in (V^{-1}_{l}(\mathrm{i}/2), \alpha_l]$ the Friedrichs extension $H^\alpha_l$ of $\widetilde H^\alpha_l$ is a restriction of
\begin{equation}\label{H_l^alpha^*}
 (\widetilde H_l^\alpha)^*= R^1\big(1 -\alpha V_{l}(\cdot\, -\mathrm i/2)\big)
\end{equation}
to
\begin{equation*}
 \mathfrak D(H^\alpha_l) =\mathfrak D^1\dot{+} \Span\big\{(\cdot\, -\mathrm i/2 +\mathrm i\zeta_{l, \alpha})^{-1}\big\},
\end{equation*}
where $\zeta_{l, \alpha}$ is the unique solution of
\begin{equation*}
 1- \alpha V_{l}(-\mathrm i\zeta_{l, \alpha}) =0
\end{equation*}
in $(-1/2, 0]$.
\end{enumerate} 
\end{lem}

We now make a crucial observation concerning the functions \eqref{extra functions on the half line} transformed in Lemma \ref{T transformed extra functions lemma}. The following lemma can be proved in the same way as Lemma 24 in \cite{MorozovMueller} using Lemma \ref{H alpha l domains lemma}.

\begin{lem}\label{domain connection lemma}
 Let $\nu\in (\sqrt{3}/2,1]$. For $s =\pm1/2$ the functions \eqref{chi}, \eqref{xi} and \eqref{eta} satisfy:
\begin{enumerate}
  \item $\xi^\nu_s$ and $\eta^\nu_s$ belong to $\mathfrak D^1$;
  \item $\chi^\nu_s$ belongs to $\mathfrak D\big(H_0^{(V_{0}(\mathrm i\Upsilon_\nu))^{-1}}\big)\cap\mathfrak D\big(H_1^{(V_{1}(\mathrm i\Upsilon_\nu))^{-1}}\big)$.
\end{enumerate}
\end{lem}

\section{Channel-wise estimates}\label{critical channels section}
We will now prove lower bounds on $|D^{\nu}_{l,s}|$ for $(l, m, s) \in\mathfrak T$ and $\nu\in [0, 1]$. For $(\nu, l, s)\in(\sqrt3/2, 1]\times \big\{(0, 1/2), (1, -1/2)\big\}$ the operator $\big(D^{\nu}_{l,s}\big)^2$ cannot be bounded from below by an operator unitarily equivalent to $R^2$, but for all other $(\nu, l,s)$ this is possible. Thus we distinguish two different types of channels.

\paragraph{Critical channels.}
For $\nu\in [0, 1]$ and $s =\pm1/2$ we introduce the $(2\times 2)$-matrix-valued function on $\mathbb R$:
\begin{equation}\label{M kappa nu}
 M_{s}^\nu:= \begin{pmatrix}-\nu V_{1/2 -s}(\cdot\, +\mathrm i/2) & 1\\ 1 & -\nu V_{1/2 +s}(\cdot\, +\mathrm i/2)\end{pmatrix}.
\end{equation}

\begin{lem}\label{maximal subdomain transformed lemma}
Let $\nu\in(\sqrt{3}/2,1]$. For $s =\pm1/2$ and any $\Psi\in\mathfrak C^\nu_{1/2 -s, s}$ there exists a decomposition
\begin{align}\label{U Psi representation}
 \mathcal U_{1/2 -s, s}\Psi =\binom{\zeta}\upsilon +a\chi_s^\nu\binom1{\nu V_{1/2 -s}(\mathrm i\Upsilon_\nu)}
\end{align}
with $\zeta, \upsilon\in \mathfrak D^1$ and $a\in \mathbb C$. 
Moreover, the representation
\begin{equation*}
\begin{split}
 \mathcal U_{1/2 -s, s}D^\nu_{1/2 -s, s}\Psi = M^\nu_{s}\binom{R^1\zeta +a\chi_s^\nu(\cdot +\mathrm i)}{R^1\upsilon +a\nu V_{1/2 -s}(\mathrm i\Upsilon_\nu)\chi_s^\nu(\cdot +\mathrm i)}
\end{split}\end{equation*}
holds.
\end{lem}

\begin{proof}
 The prove follows the same lines as the one of Lemma 24 in \cite{MorozovMueller}.
 The decomposition \eqref{U Psi representation} follows from \eqref{operator core}, Lemma \ref{T transformed extra functions lemma}, and Lemma \ref{domain connection lemma}.
 For any $(\varpi, \varsigma)\in \mathsf C^\infty_0\big(\mathbb R_+, \mathbb C^2\big)$ using \eqref{D nu}, \eqref{Upsi extra}, Lemmata \ref{transformed operators lemma} and \ref{1/r with T lemma}, Lemma \ref{V corollary} and \eqref{H_l^alpha^*} we obtain
\begin{equation*}
\begin{split}
 &\Big\langle D^\nu_{1/2 -s, s}\Psi, \binom{\varpi}{\varsigma}\Big\rangle =\Big\langle \mathcal U_{1/2 -s, s}^*M^\nu_{s}\binom{R^1\zeta +a\chi_s^\nu(\cdot +\mathrm i)}{R^1\upsilon +a\nu V_{1/2 -s}(\mathrm i\Upsilon_\nu)\chi_s^\nu(\cdot +\mathrm i)}, \binom{\varpi}{\varsigma}\Big\rangle.
\end{split}\end{equation*}
By density of $\mathsf C_0^\infty\big(\mathbb R_+, \mathbb C^2)$ the claim follows.
\end{proof}

The proof of the following lemma can be found in Appendix \ref{a: eta_nu}.

\begin{lem}\label{l: eta_nu}
The map $\nu\mapsto \eta_\nu$ given by \eqref{eta_nu} is real-analytic and mo\-no\-to\-no\-us\-ly decreasing on $[0,1]$.
\end{lem}

The next lemma is the most technically demanding.

\begin{lem}\label{critical Dirac via Kato bound lemma}
For $\nu\in [0, 1]$ define the function
\begin{align}\label{K}
 K^\nu:= \Big| 1 - \big(V_{0}(\mathrm i\Upsilon_\nu)\big)^{-1}V_{0}(\cdot +\mathrm i/2)\Big|^2
\end{align}
on $\mathbb R$.
Then for $s =\pm1/2$ the maximal value of $\eta_\nu$ for which the lower bound
\begin{align}\label{critical Dirac via Kato bound}
 (M_{s}^\nu)^*M_{s}^\nu \geqslant \eta_{\nu}^2 K^{\nu}\otimes \mathds{1}_{\mathbb{C}^2}
\end{align}
holds point-wise on $\mathbb R$ is given by \eqref{eta_nu}.
\end{lem}

\begin{proof}
It is enough to study the case of $M_{1/2}^\nu$ in \eqref{critical Dirac via Kato bound}, since the case of $M_{-1/2}^\nu$ follows immediately from the relation $M_{-1/2}^\nu =\sigma_1M_{1/2}^\nu\sigma_1$.
For $\tau\in \mathbb R$ consider
\begin{align*}
P(\tau) := \frac{\Gamma\big((3- 2\mathrm i\tau)/4\big)\Gamma\big((1+ 2\mathrm i\tau)/4\big)}{\Gamma\big((3+ 2\mathrm i\tau)/4\big)\Gamma\big((1- 2\mathrm i\tau)/4\big)}.
\end{align*}
For all $\tau\in \mathbb R$ we have
\begin{align}\label{mod P}
 \big|P(\tau)\big| =1 
\end{align}
and (see 5.4.5, 4.35.11, 4.35.16, 4.35.19, 4.28.4 and 4.28.6 in \cite{dlmf})
\begin{align}\label{trig P}
 P(\tau) =\sech(\pi\tau) -\mathrm i\tanh(\pi\tau).
\end{align}
By \eqref{V} we have
\begin{align*}
 V_0(\tau +\mathrm i/2)= \frac{2}{1- 2\mathrm i\tau}P(\tau), \quad V_1(\tau- \mathrm i/2) =\frac{2(1+ 2\mathrm i\tau)}{(1- 2\mathrm i \tau)(3 +2\mathrm i\tau)}P(\tau).
\end{align*}
Thus according to \eqref{M kappa nu} for $\tau\in \mathbb R$ we have
\begin{align}\label{M*M}
 \big(M_{1/2}^\nu(\tau)\big)^*M_{1/2}^\nu(\tau) =\begin{pmatrix}
                                                  1 +\dfrac{4\nu^2}{1 +4\tau^2} & \dfrac{-8\nu(1 -\mathrm i\tau)\overline{P(\tau)}}{(1 +2\mathrm i\tau)(3- 2\mathrm i\tau)}\\ \dfrac{-8\nu(1 +\mathrm i\tau)P(\tau)}{(1 -2\mathrm i\tau)(3 +2\mathrm i\tau)} & 1 +\dfrac{4\nu^2}{9 +4\tau^2}
                                                 \end{pmatrix}.
\end{align}

In the following the function $z\cot z$ is analytically continued to $z =0$, i.e. $z\cot z\arrowvert_{z =0} :=1$.
Representations \eqref{K}, \eqref{mod P} and \eqref{trig P} imply
\begin{align}\label{K rewritten}
 K^\nu(\tau) =\frac{1+ 4\tau^2 +4\Upsilon_\nu^2\cot^2\Big(\dfrac{\pi\Upsilon_\nu}2\Big) -4\Upsilon_\nu\cot\Big(\dfrac{\pi\Upsilon_\nu}2\Big)\big(\sech(\pi\tau) +2\tau\tanh(\pi\tau)\big)}{1 +4\tau^2}.
\end{align}
Inequality \eqref{critical Dirac via Kato bound} is equivalent to the non-negativity of the smallest eigenvalue of
\begin{align*}
 \big(M_{1/2}^\nu(\tau)\big)^*M_{1/2}^\nu(\tau) -\eta_{\nu}^2 K^{\nu}(\tau)\otimes \mathds{1}_{\mathbb{C}^2}.
\end{align*}
Using \eqref{M*M} and \eqref{K rewritten} we find this eigenvalue to be given by
\begin{align}\label{lambda}
 \lambda(\tau) =1 -\eta_\nu^2 K^\nu(\tau) +\frac{4\nu^2(5 +4\tau^2) -8\nu\big(4\nu^2 +(1 +\tau^2)(1+ 4\tau^2)(9 +4\tau^2)\big)^{1/2}}{(1+ 4\tau^2)(9 +4\tau^2)}.
\end{align}
Now it is enough to prove the non-negativity of \eqref{lambda} with $\eta_\nu$ and $K^\nu$ given by \eqref{eta_nu} and \eqref{K rewritten}, respectively. Using \eqref{K rewritten} and \eqref{eta_nu} we observe that the right hand side of \eqref{lambda} vanishes at $\tau =0$, i.e. the value of $\eta_\nu$ cannot be increased.

Multiplying \eqref{lambda} by $(1+ 4\tau^2)(9 +4\tau^2)$ and rewriting the result using \eqref{K rewritten} we conclude that it is enough to prove the non-negativity of
\begin{align}\begin{split}\label{p}
 p(\tau, \nu) &:=(1 -\eta_\nu^2)(1+ 4\tau^2)(9 +4\tau^2) +4\nu^2(5 +4\tau^2)\\ &+4\Upsilon_\nu\eta_\nu^2(9 +4\tau^2)\big(\sech(\pi\tau) +2\tau\tanh(\pi\tau)\big)\cot(\pi\Upsilon_\nu/2)\\ &-4\Upsilon_\nu^2\eta_\nu^2(9 +4\tau^2)\cot^2(\pi\Upsilon_\nu/2)\\ &-8\nu\big(4\nu^2 +(1 +\tau^2)(1+ 4\tau^2)(9 +4\tau^2)\big)^{1/2}
\end{split}\end{align}
for all $(\tau, \nu)\in [0, \infty)\times [0, 1]$.

Observing that
\begin{align*}
 (9 +4\tau^2)\big(1 +2\tau\sinh(\pi\tau)\big)-\big(9 +(4 +18\pi -9\pi^2/2)\tau^2\big)\cosh(\pi\tau) \geqslant0
\end{align*}
holds for all $\tau\in\mathbb R$ by positivity of all the coefficients in the Taylor expansion at $\tau =0$,
we conclude
\begin{align*}
 (9 +4\tau^2)\big(\sech(\pi\tau) +2\tau\tanh(\pi\tau)\big) \geqslant 9 +(4 +18\pi -9\pi^2/2)\tau^2.
\end{align*}
Thus to prove the non-negativity of \eqref{p} it suffices to establish the inequality
\begin{align}\label{ppm}
 p_+(\tau, \nu) \geqslant p_-(\tau, \nu)
\end{align}
with
\begin{align*}
 p_+(\tau, \nu) &:=(1 -\eta_\nu^2)(1+ 4\tau^2)(9 +4\tau^2) +4\nu^2(5 +4\tau^2)\\ &+4\Upsilon_\nu\eta_\nu^2\big(9 +(4 +18\pi -9\pi^2/2)\tau^2\big)\cot(\pi\Upsilon_\nu/2)\\ &-4\Upsilon_\nu^2\eta_\nu^2(9 +4\tau^2)\cot^2(\pi\Upsilon_\nu/2)
\end{align*}
and
\begin{align*}
 p_-(\tau, \nu) :=8\nu\big(4\nu^2 +(1 +\tau^2)(1+ 4\tau^2)(9 +4\tau^2)\big)^{1/2}.
\end{align*}
For $x\in [0, 1]$ we have
\begin{align*}
 x^2\cot^2(\pi x/2) \leqslant x\cot(\pi x/2) <1,
\end{align*}
and, since Lemma \ref{l: eta_nu} implies $\eta_\nu\leqslant \eta_0 =1$, we get $p_+(\tau, \nu) \geqslant0$. Thus \eqref{ppm} is equivalent to
\begin{align}\label{ps squared}
 \big(p_+(\tau, \nu)\big)^2 -\big(p_-(\tau, \nu)\big)^2 =c_2(\nu)\tau^2 + c_4(\nu)\tau^4 +c_6(\nu)\tau^6 +c_8(\nu)\tau^8 \geqslant0
\end{align}
with
\begin{align}
 c_2(\nu) &:=\frac{16\nu\sqrt{9 +4\nu^2}}9\Bigg(72(5 +2\nu^2)-\frac{1764\nu}{\sqrt{9 +4\nu^2}} + (\sqrt{9 +4\nu^2} -4\nu)^2 \notag\\ &\times\bigg(-4+\frac{9\pi^2(\pi- 4)^2}{4(\pi -2)^2} -\Big(\frac{3\pi(4 -\pi)}{2(\pi -2)} +\frac{3(\pi -2)}{1 -2\Upsilon_\nu\cot(\pi\Upsilon_\nu/2)}\Big)^2\bigg)\Bigg),\label{c2}\\
 c_4(\nu) &:=\frac{\big(c_2(\nu)+3136\nu^2\big)^2}{256\nu^2(9 +4\nu^2)} -3584\nu^2 +256\nu\sqrt{9 +4\nu^2}(1 -\eta_\nu^2),\label{c4}\\
 c_6(\nu) &:=\frac{2\big(c_2(\nu)+3136\nu^2\big)}{\nu\sqrt{9 +4\nu^2}}(1 -\eta_\nu^2) -1024\nu^2\notag\\
\intertext{and}
 c_8(\nu) &:=256(1 -\eta_\nu^2)^2.\notag
\end{align}
 In order to establish \eqref{ps squared} it is enough to observe that for all $j= 2, 4, 6, 8$ the functions $c_j$ are non-negative on $[0, 1]$. Taylor series expansions show that $\nu =0$ is a local minimum with value $0$ for all of these functions, whereas $\nu =\sqrt3/2$ is another local minimum with value $0$ for $c_2$.

A rigorous proof of the non-negativity of $c_2$, $c_4$ and $c_6$ can be found in Appendix \ref{a: Dirac via Kato remainder}.
\end{proof}

The following lemma is analogous to Lemma 28 in \cite{MorozovMueller}.
\begin{lem}\label{moduli comparison lemma}
The inequality
\begin{align}\label{moduli comparison}
\begin{split}
(D_{1/2 -s, s}^\nu)^2 \geqslant \eta_\nu^2\big(\mathcal U_{1/2 -s, s}^*(H_{0}^{(V_{0}(\mathrm i\Upsilon_\nu))^{-1}}\otimes \mathds{1}_{\mathbb{C}^2})\mathcal U_{1/2 -s, s}\big)^2
\end{split}
\end{align}
holds for any $\nu\in [0, 1]$ and $s= \pm1/2$.
\end{lem}

\begin{proof}
For $\nu\in(\sqrt 3/2, 1]$ and arbitrary $\Psi\in\mathfrak C^\nu_{1/2 -s, s}$ we use \eqref{U Psi representation} to represent $\mathcal U_{s}\Psi$.
Applying Lemmata \ref{maximal subdomain transformed lemma}, \ref{critical Dirac via Kato bound lemma}, \ref{domain connection lemma} together with Lemma \ref{H alpha l domains lemma}$(iii)$ we get
\begin{align*}
 &\|D_{1/2 -s, s}^\nu\Psi\|^2 = \bigg\|M^\nu_{s}\binom{R^1\zeta +a\chi_s^\nu(\cdot +\mathrm i)}{R^1\upsilon +a\nu V_{1/2 -s}(\mathrm i\Upsilon_\nu)\chi_s^\nu(\cdot +\mathrm i)}\bigg\|^2\\ &\geqslant \eta_\nu^2\left\|\Big(1 -\frac{V_0(\cdot +\mathrm i/2)}{V_0(\mathrm i\Upsilon_\nu)}\Big)\begin{pmatrix}R^1\zeta +a\chi_s^\nu(\cdot +\mathrm i)\\ R^1\upsilon +a\nu V_{1/2 -s}(\mathrm i\Upsilon_\nu)\chi_s^\nu(\cdot +\mathrm i)\end{pmatrix}\right\|^2\\
&=\eta_\nu^2\big\|\mathcal U_{1/2 -s, s}^*(H_{0}^{(V_{0}(\mathrm i\Upsilon_\nu))^{-1}}\otimes \mathds{1}_{\mathbb{C}^2})\mathcal U_{1/2 -s, s}\Psi\big\|^2.
\end{align*}
Since $\mathfrak C^\nu_{1/2 -s, s}$ is an operator core for $D_{1/2 -s, s}^\nu$, we conclude \eqref{moduli comparison}.\\
For $\nu \in [0, \sqrt3/2]$, according to \eqref{operator core} we can do the same computation with $a :=0$ using parts $(i)$ and $(ii)$ of Lemma \ref{H alpha l domains lemma} instead of part $(iii)$.
\end{proof}

\paragraph{Non-critical channels.}

\begin{lem}\label{non-critical channels lemma}
Let $(l,m,s)\in\mathfrak{T}$ such that $(l,s)\notin \big\{(0,1/2),(1,-1/2)\big\}$. Then the operator inequalities
\begin{align*}
 (D_{l,s}^\nu)^2 \geqslant \frac{\big((225 +4\nu^2)^{1/2}-8\nu\big)^2}{225}\,(\mathcal U_{l,s}^*R^1\mathcal U_{l,s})^2
\end{align*}
hold for all $\nu\in [0, 1]$.
\end{lem}

\begin{proof}
Analogously to the proof of Lemma 29 in \cite{MorozovMueller} we conclude that the inequality
\begin{align*}
 (D_{l,s}^\nu)^2 \geqslant (1 -b)\,(\mathcal U_{l,s}^*R^1\mathcal U_{l,s})^2
\end{align*}
with $b\in [0, 1]$ follows from the non-negativity of the function
\begin{align}\label{a function}
  a^\nu_{\varkappa, -}(b, \tau):= \nu^2 +b/4 +\varkappa^2b +\tau^2b -(4\varkappa^2\nu^2 +4\nu^2\tau^2 +\varkappa^2b^2)^{1/2},
\end{align}
with $\varkappa := 2sl +s +1/2\in \mathbb Z\setminus\{-1, 0, 1\}$ for all $\tau\in\mathbb R$. Again as in the proof of Lemma 29 in \cite{MorozovMueller}, we observe that the function \eqref{a function} is even in both $\varkappa$ and $\tau$ and monotonously increasing in $\varkappa$ for $\varkappa \geqslant 2$, $\tau \in\mathbb R$ provided $b \geqslant \nu/\sqrt6$. Furthermore, $a^\nu_{2, -}(b, \cdot)$ is monotonously increasing on $[0, \infty)$ provided $b \geqslant(\sqrt 5 -2)^{1/2}\nu$. It remains to find the smallest $b \in[0, 1]$ satisfying the above assumptions together with $a^\nu_{2, -}(b, 0) \geqslant 0$.
\end{proof}

\section{Proofs of the main theorems}\label{main proofs section}
\subsection{Proof of Theorem \ref{absolute value bounds theorem}}

In the same way as in the proof of Theorem 1 in \cite{MorozovMueller}, we obtain from Lemmata \ref{angular Dirac lemma}, \ref{critical constants lemma}, \ref{moduli comparison lemma}, \ref{non-critical channels lemma} and \ref{transformed operators lemma} the inequality \eqref{absolute value bound} with
\begin{align}\label{C_nu preliminary}
 C_\nu :=\min\bigg\{\eta_\nu\Big(1 -\frac{V_0(0)}{V_0(\mathrm i\Upsilon_\nu)}\Big), \frac1{15}\big(\sqrt{225 +4\nu^2} -8\nu\big)\bigg\}.
\end{align}
Using \eqref{V} we observe the identity
\begin{align*}
 \eta_\nu\Big(1 -\frac{V_0(0)}{V_0(\mathrm i\Upsilon_\nu)}\Big) =\frac\pi{12}(\sqrt{9 + 4 \nu^2}-4 \nu) +(1 -\pi/4)\eta_\nu.
\end{align*}
Hence by Lemma \ref{l: eta_nu} and concavity of the square root we obtain
\begin{align}\label{min upper bound}
 \eta_\nu\Big(1 -\frac{V_0(0)}{V_0(\mathrm i\Upsilon_\nu)}\Big)\leqslant \frac\pi{12}\big(3 (1 +2\nu^2/9) -4\nu\big) +(1 -\pi/4).
\end{align}
Estimating
\begin{align*}
 \frac1{15}\big(\sqrt{225 +4\nu^2} -8\nu\big) \geqslant \frac1{15}\big(\sqrt{225} -8\nu\big)\geqslant \frac\pi{12}\big(3 (1 +2\nu^2/9) -4\nu\big) +(1 -\pi/4)
\end{align*}
for $\nu\in [0, 1]$ and comparing to \eqref{min upper bound} we conclude that the minimum in \eqref{C_nu preliminary} is always achieved at the first entry, i.e. \eqref{c nu} holds.

\subsection{Proof of Corollary \ref{c: massive estimates}}

For arbitrary $M >0$ the explicit solution to the eigenvalue problem (see e.g. (7.118) in \cite{Thaller}) provides the lower bound
\begin{align}\label{spectral bound}
 \big(D^{\nu, M}\big)^2 \geqslant M^2(1 -\nu^2).
\end{align}
It implies, in its turn,
\begin{align*}\begin{split}
 \big(D^{\nu, M}\big)^2 =\Upsilon_\nu^2(D^\nu)^2 +\nu^2\big(D^{\nu, M/\nu^2}\big)^2 +M^2(1 -\nu^{-2}) \geqslant \Upsilon_\nu^2(D^\nu)^2.
\end{split}\end{align*}
By operator monotonicity of the square root and Theorem \ref{absolute value bounds theorem} we conclude \eqref{massive modulus}.

Taking a convex combination of \eqref{massive modulus} and the square root of \eqref{spectral bound}, for any $\theta\in [0, 1]$ we get
\begin{align*}
 \big|D^{\nu ,M}\big| \geqslant \Upsilon_\nu \big(\theta C_\nu(-\Delta)^{1/2} +(1 -\theta)M\big).
\end{align*}
Computing
\begin{align*}
 \max_{\theta\in [0, 1]}\ \inf_{p\in [0, \infty)}\frac{\theta C_\nu p +(1 -\theta)M}{\sqrt{p^2 +M^2}}
\end{align*}
we arrive at
\begin{align*}
  |D^{\nu, M}| &\geqslant \frac{\Upsilon_\nu C_\nu}{1 +C_\nu}|D^{0, M}|.
\end{align*}
Combining this with the estimate
\begin{align*}
  |D^{\nu, M}| \geqslant (1 -2\nu)|D^{0, M}|,
\end{align*}
which follows immediately from the Hardy inequality (see Lemma 1 in \cite{Bach1999}), we obtain \eqref{massive bound 2}.

\subsection{Proof of Theorem \ref{t: l bound}}

The proof is analogous to the one of part 2 of Theorem 1 in \cite{MorozovMueller}. It is enough to combine Lemmata \ref{moduli comparison lemma}, \ref{non-critical channels lemma}, \ref{critical lower bound lemma} with the inequality
\begin{align*}
  R^1 \geqslant \lambda^{-\lambda}(1 -\lambda)^{\lambda -1} a^{\lambda -1}R^\lambda -a^{-1}
\end{align*}
valid for all $a, R\in [0, \infty)$.

\subsection{Proof of Theorem \ref{t: virtual level}}

The proof is analogous to the one of Part I of Theorem 2.5 in \cite{morozov2016virtual}.

Let $\mathbf j \in \{-1, 1\}^{2}$ be such that the conditions \eqref{V integrability} and \eqref{VL condition} are satisfied.
For any $\psi\in P_{[0, \infty)}(D_{(j_2 +1)/2, -j_2/2}^1) \mathfrak D(D_{(j_2 +1)/2, -j_2/2}^1)$ we have
\begin{align*}
\Psi :=\mathcal A^* \,\bigoplus_{(l, m, s)\in \mathfrak T}\delta_{l, j_2/2 +1/2}\delta_{m, j_1/2}\delta_{s, -j_2/2}\psi \in P_{+}^1\mathfrak D(D^1)
\end{align*}
and a calculation based upon \eqref{D nu}, \eqref{A}, \eqref{spherical series}, (2.1.25) and (2.1.26) in \cite{BalinskyEvans} delivers
\begin{align}\label{channel passage}
 \langle\Psi, D^1(V)\Psi\rangle_{\mathsf L^2(\mathbb R^3, \mathbb C^4)} =\langle\psi, (D_{(j_2 +1)/2, -j_2/2}^1 -W_{\mathbf j}/4\pi)\psi\rangle_{\mathsf L^2(\mathbb R_+, \mathbb C^2)}.
\end{align}
Now Part I of Theorem 2.3 in \cite{morozov2016virtual} guarantees the existence of $\psi$ such that \eqref{channel passage} is negative. Hence by the minimax principle $D^1(V)$ has non-empty negative spectrum.
According to Theorem \ref{CLR theorem}(b) the negative spectrum of $D^1(V)$ can only consist of eigenvalues.

\appendix

\section{Proof of Lemma \ref{l: eta_nu}}\label{a: eta_nu}

Both $g(\nu):= (9+4\nu^2)^{1/2}-4\nu$ and $h(\nu):= 3\big(1-2\Upsilon_\nu\cot(\pi\Upsilon_\nu/2)\big)$ (with the convention $z\cot z\arrowvert_{z =0} :=1$) are monotonously decreasing real-analytic functions of $\nu\in [0,1]$ taking value $3$ at $\nu =0$ and having simple zeroes at $\nu= \sqrt3/2$.
Thus the real analyticity of $\eta_\nu$ follows from such of the numerator and denominator in \eqref{eta_nu} and the cancellation of zeroes. In the following it is enough to consider $\nu\in (0, 1)$ and extend the results by continuity.

The monotonicity of $\nu\mapsto \eta_\nu$ follows from convexity of $g$ and concavity of $h$ together with the fact that both functions are monotonously decreasing and vanish at $\nu= \sqrt3/2$. Indeed, then $\big|g(\nu)\big|/|\nu -\sqrt3/2|$ is monotonously decreasing and $\big|h(\nu)\big|/|\nu -\sqrt3/2|$ is monotonously increasing. Thus for every $\nu_1 <\nu_2\in (0, 1)\setminus\{\sqrt3/2\}$ we have
\begin{align*}
 \frac{g(\nu_2)}{h(\nu_2)} = \frac{\big|g(\nu_2)\big|}{\big|h(\nu_2)\big|} \leqslant \frac{\big|g(\nu_1)\big|}{\big|h(\nu_1)\big|} =\frac{g(\nu_1)}{h(\nu_1)}.
\end{align*}
The claimed convexity of $g$ and concavity of $h$ can be observed by studying their second derivatives.
We have $g''(\nu) =36(9 +4\nu^2)^{-3/2} >0$ and $h''(\nu) =3\big(f_1(\nu) -f_2(\nu)\big)$ with
\begin{align}\label{f_1}
\begin{split}f_1(\nu) &:=\frac{\nu^2\big(\pi\Upsilon_\nu\sin(\pi\Upsilon_\nu/2) -\pi^2\Upsilon_\nu^2\cos(\pi\Upsilon_\nu/2) +2\cos(\pi\Upsilon_\nu/2)\sin^2(\pi\Upsilon_\nu/2)\big)}{\Upsilon_\nu^3\sin^3(\pi\Upsilon_\nu/2)}\\ &< 2\end{split}
\end{align}
and
\begin{align}\label{f_2}
 f_2(\nu) :=\frac{\pi\Upsilon_\nu -\sin(\pi\Upsilon_\nu)}{\Upsilon_\nu\sin^2(\pi\Upsilon_\nu/2)} >2.
\end{align}
To justify \eqref{f_1}, for $z\in[0, \pi/2]$ we can combine the elementary estimates 
\begin{align*}
 2z\sin z &\leqslant 2z^2 -\frac{z^4}3 +\frac{z^6}{60},\\
 -4z^2\cos z &\leqslant -4z^2 +2z^4 -\frac{z^6}6 +\frac{z^8}{180},\\
 2(\cos z)\sin^2z =(\cos z)\big(1 -\cos(2z)\big) &\leqslant \Big(1 -\frac{z^2}2 +\frac{z^4}{24}\Big)\Big(2z^2 -\frac{2z^4}3 +\frac{4z^6}{45}\Big)
\end{align*}
which follow from the Taylor expansions of $\sin z$ and $\cos z$ at $z =0$ to conclude
\begin{align}\begin{split}\label{top}
 2z\sin z -4z^2\cos z +2(\cos z)\sin^2z &\leqslant \frac{16}{45}z^6 -\frac{z^8}{15} +\frac{z^{10}}{270} \\ &\leqslant \frac{16}{45}z^6 +\Big(-\frac{1}{15} +\frac{\pi^2}{4\cdot 270}\Big)z^8 \leqslant \frac{16}{45}z^6.
\end{split}\end{align}
Analogously, for $z\in[0, \pi/2]$ we get
\begin{align}\label{bottom}
 \sin^3z \geqslant \Big(z -\frac{z^3}6\Big)^3 \geqslant z^3 -\frac{z^5}2 +\Big(\frac1{12} -\frac{\pi^2}{4\cdot 216}\Big)z^7.
\end{align}
Multiplying \eqref{f_1} by the denominator and applying the inequalities \eqref{top} and \eqref{bottom} we conclude its validity.
Inequality \eqref{f_2} for $z\in [0, \pi/2]$ follows analogously from
\begin{align*}
 2z -\sin(2z) &\geqslant \frac{4z^3}3 -\frac{4z^5}{15}\quad\text{and}\\
 \sin^2z = \frac{1- \cos(2z)}2 &\leqslant z^2 -\frac{z^4}3 +\frac{2z^6}{45}.
\end{align*}

\section{Remainder of the proof of Lemma \ref{critical Dirac via Kato bound lemma}: the non-negativity of \texorpdfstring{$c_2$, $c_4$ and $c_6$}{c2, c4 and c6}}\label{a: Dirac via Kato remainder}

We first prove appropriate estimates on the function $z\mapsto z\cot z$.
\begin{lem}\ \label{l: cot bounds}
\begin{enumerate}
  \item For $n\in\mathbb N$ let
\begin{align*}
 F_n(\nu) :=\sum_{k =1}^n(-1)^{k +1}\bigg(\sum_{j =1}^\infty \frac{16j^2}{\pi(4j^2 -1)^{k +1}}\bigg)\nu^{2k}.
\end{align*}
Then for all $\nu\in[0, 1]$ and $n\in\mathbb N$ the inequalities
\begin{align}\label{cot via F}
 F_{2n}(\nu) \leqslant \Upsilon_\nu\cot(\pi\Upsilon_\nu/2) \leqslant F_{2n -1}(\nu)
\end{align}
hold.
 \item\label{bound 2} For $n\in\mathbb N$ let
\begin{align}\label{S_n and T_n}
 S_n :=\frac{2^{2n +2}}\pi\sum_{k =0}^\infty\frac{(-1)^k}{(2k +1)^{2n +1}},\quad\text{and}\quad T_n :=\frac{2^{2n +1}}\pi\sum_{k =0}^\infty\frac{1}{(2k +1)^{2n}}.
\end{align}
Then for all $x \in(-1/2, 0]$ and $n\in\mathbb N$ the inequality
\begin{align}\label{cot lower bound}
 \Big(\frac12 +x\Big)\cot\Big(\frac{\pi(1/2 +x)}2\Big) \geqslant \frac12 +\sum_{j =1}^n\Big(\frac{S_j}2 -T_j\Big)x^{2j} +\Big(S_{j -1} -\frac{T_j}2\Big)x^{2j -1}
\end{align}
holds.
\item\label{bound 3} For all $x\in[0, 1/2)$ and $l, k\in\mathbb N$ the inequality
\begin{align}\label{cot upper bound}
 \Big(\frac12 +x\Big)\cot\Big(\frac{\pi(\frac12 +x)}2\Big) \leqslant \frac12 +\sum_{j =1}^l\Big(\frac{S_j}2 -T_j\Big)x^{2j} +\sum_{j =1}^k\Big(S_{j -1} -\frac{T_j}2\Big)x^{2j -1}
\end{align}
holds true.
\end{enumerate}
\end{lem}

\begin{proof}
\emph{1.} Applying \eqref{beta}, the series representation 4.22.3 in \cite{dlmf} and using the geometric series we arrive at
\begin{align}\label{beta cot beta}
  \Upsilon_\nu\cot(\pi\Upsilon_\nu/2) =\frac2\pi -\frac4\pi\sum_{n =1}^\infty\frac1{4n^2 -1} +\sum_{k =1}^\infty(-1)^{k +1}\bigg(\sum_{n =1}^\infty\frac{16n^2}{\pi(4n^2 -1)^{1 +k}}\bigg)\nu^{2k}.
\end{align}
By setting $\nu :=0$ in \eqref{beta cot beta} we obtain
\begin{align*}
 \frac2\pi -\frac4\pi\sum_{n =1}^\infty\frac1{4n^2 -1} =0.
\end{align*}
It remains to observe that the last series in \eqref{beta cot beta} is alternating with monotonously decreasing absolute values of the terms.\\
\emph{2, 3.} By 4.19.5, 24.8.4 and 24.2.9 in \cite{dlmf} we obtain 
\begin{align}\label{sec series}
 \sec(\pi x) =\sum_{n =0}^\infty S_nx^{2n} \quad\text{for }|x| <1/2.
\end{align}
By 4.19.3, 25.6.2, 24.2.2 and 25.2.2 in \cite{dlmf} we have
\begin{align}\label{tan series}
 \tan(\pi x) =\sum_{n =1}^\infty T_nx^{2n -1} \quad\text{for }|x| <1/2.
\end{align}
Applying the trigonometric identity
\begin{align*}
 \cot\big(\pi(1/2 +x)/2\big) =\sec(\pi x) -\tan(\pi x)
\end{align*}
we arrive at
\begin{align}\label{cot series}
 (1/2 +x)\cot\big(\pi(1/2 +x)/2\big) =\frac12 +\sum_{j =1}^\infty\bigg(\Big(\frac{S_j}2 -T_j\Big)x +\Big(S_{j -1} -\frac{T_j}2\Big)\bigg)x^{2j -1}.
\end{align}
Since \eqref{sec series} and \eqref{tan series} imply
\begin{align}\label{ST differences}
 \frac{S_j}2 -T_j <0\quad\text{and}\quad S_{j -1} -\frac{T_j}2 <0\quad\text{for all }j\in\mathbb N,
\end{align}
\eqref{cot upper bound} with $x\in [0, 1/2)$ follows from \eqref{cot series}.

For $x \in(-1/2, 0]$ using \eqref{ST differences} and \eqref{S_n and T_n} we conclude for $j\in\mathbb N$
\begin{align}\label{ST negativity}
\begin{split}
 &\Big(\frac{S_j}2 -T_j\Big)x +\Big(S_{j -1} -\frac{T_j}2\Big) \leqslant \frac12\Big(T_j -\frac{S_j}2\Big) +S_{j -1} -\frac{T_j}2\\ &=S_{j -1} -\frac{S_j}4 =\frac{2^{2j}}\pi\bigg(\sum_{k =1}^\infty(-1)^k\frac{(2k +1)^2 -1}{(2k +1)^{2j +1}}\bigg) <0.
\end{split}\end{align}
Together with \eqref{cot series} and $x^{2j -1}\leqslant 0$ this gives the lower bound \eqref{cot lower bound}.
\end{proof}

\paragraph{Non-negativity of $c_2$ for $\nu\in [0, 3/5]$.} Let $b_1(\nu)$ be given by the right hand side of \eqref{c2}, but with $\pi\nu^2/2$ instead of $2\Upsilon_\nu\cot(\pi\Upsilon_\nu/2)$. The inequality
\begin{align}\label{bound 1}
 1 -2\Upsilon_\nu\cot(\pi\Upsilon_\nu/2) \geqslant 1 -\pi\nu^2/2
\end{align}
which follows from \eqref{cot via F} implies that $b_1(\nu)\leqslant c_2(\nu)$ holds for all $\nu$ for which the right hand side of \eqref{bound 1} remains positive, thus for all $\nu\in [0, 3/5]$. Hence the non-negativity of $c_2(\nu)$ follows from the non-negativity of
\begin{align}\label{via p1q1}
 \frac{9(\pi\nu^2 -2)^2}{32\nu^2}b_1(\nu) =p_1(\nu) +q_1(\nu),
\end{align}
where
\begin{align*}
 p_1(\nu) &:=2232 + (2560 +2952\pi -2592\pi^2 +648\pi^3)\nu^2\\ &+(288\pi^3 -256\pi -1890\pi^2)\nu^4 +64\pi^2\nu^6 \geqslant 0\text{ for all }\nu\in [0, 3/5]
\end{align*}
and
\begin{align*}
 q_1(\nu) &:=\big((324\pi^2 -1312 -648\pi -81\pi^3)\nu\\& +(882\pi^2 -128\pi -180\pi^3)\nu^3 +32\pi^2\nu^5\big)\sqrt{9 +4\nu^2}.
\end{align*}
To establish the non-negativity of \eqref{via p1q1} it is thus enough to observe the non-negativity of
\begin{align*}
 p_2(\nu) :=p_1^2(\nu) -q_1^2(\nu),
\end{align*}
which is a polynomial of degree $5$ in $\nu^2$. Such non-negativity for $\nu\in [0, 3/5]$ follows immediately from the fact that
\begin{align*}
 p_3(y) :=p_2\big(\sqrt{9/25 -y}\big)
\end{align*}
is a polynomial with positive coefficients and is thus positive for all $y\in[0, 9/25]$.

\paragraph{Non-negativity of $c_2$ for $\nu\in [3/5, 1]$.} Substituting $\nu :=\sqrt{1 -(1/2 +x)^2}$ with $x\in [-1/2, 1/2]$ and introducing
\begin{align*}
 r_1(x):= \sqrt{12 -4x -4x^2}\quad\text{and}\quad r_2(x):=\sqrt{3 -4x -4x^2}
\end{align*}
we conclude that the non-negativity of $c_2(\nu)$ for $\nu\in [0, 1]$ is equivalent to the non-negativity of
\begin{align*}\begin{split}
 &d_2(x) :=\frac{9c_2\big(\sqrt{1 -(1/2 +x)^2}\big)}{16\sqrt{1 -(1/2 +x)^2}}\\
 &=r_1(x)\Bigg(36(13 -4x -4x^2) -882\frac{r_2(x)}{r_1(x)} +\big(r_1(x) -2r_2(x)\big)^2\\ &\times\bigg(\frac{9\pi^2(4 -\pi)^2}{4(\pi -2)^2} -4 -\Big(\frac{3\pi(4 -\pi)}{2(\pi -2)} +\frac{3(\pi -2)}{1 -2(1/2 +x)\cot\big(\pi(1/2 +x)/2\big)}\Big)^2\bigg)\Bigg)
\end{split}\end{align*}
for $x\in [-1/2, 1/2]$.
Introducing for $x\in [-1/2, 1/2]$ the functions
\begin{align*}
 h(x) &:=1 +(2 -\pi)x + \Big(\frac{\pi^2}2 -2\pi\Big)x^2 + \Big(\pi^2 -\frac{\pi^3}3\Big)x^3 + \Big(\frac{5\pi^4}{24} -\frac{2\pi^3}3\Big)x^4
\intertext{and}
 l(x) &:=2(1/2 +x)\cot\big(\pi(1/2 +x)/2\big)
\end{align*}
by Lemma \ref{l: cot bounds}.\ref{bound 3} and \eqref{ST differences} we obtain
\begin{align}\label{hl bound}
 1 -l(x) \geqslant 1 -h(x) \geqslant 0\text{ for }x\in[0, 1/2].
\end{align}
On the other hand, Lemma \ref{l: cot bounds}.\ref{bound 2}, \eqref{ST negativity} and the monotonicity of $y\mapsto y\cot y$ on $[0, \pi/2]$ imply
\begin{align}\label{lh bound}
 1 -l(-1/2) \leqslant 1 -l(x) \leqslant 1 -h(x) \leqslant0\text{ for }x\in[-1/2, 0].
\end{align}
Combining \eqref{hl bound} and \eqref{lh bound} we obtain

\begin{align*}
 \Big(\frac{3\pi(4 -\pi)}{2(\pi -2)} +\frac{3(\pi -2)}{1 -l(x)}\Big)^2 \leqslant \Big(\frac{3\pi(4 -\pi)}{2(\pi -2)} +\frac{3(\pi -2)}{1 -h(x)}\Big)^2
\end{align*}
and thus 
\begin{align*}\begin{split}
 b_2(x) &:=r_1(x)\Bigg(36(13 -4x -4x^2) -882\frac{r_2(x)}{r_1(x)} +\big(r_1(x) -2r_2(x)\big)^2\\ &\times\bigg(\frac{9\pi^2(4 -\pi)^2}{4(\pi -2)^2} -4 -\Big(\frac{3\pi(4 -\pi)}{2(\pi -2)}\\ &+\frac{3(\pi -2)}{(\pi -2)x + (2\pi -\frac{\pi^2}2)x^2 + (\frac{\pi^3}3 -\pi^2)x^3 + (\frac{2\pi^3}3 -\frac{5\pi^4}{24})x^4}\Big)^2\bigg)\Bigg)
\end{split}\end{align*}
is a lower bound on $d_2(x)$ for all $x\in[-1/2, 1/2]$. To prove the non-negativity of $c_2(\nu)$ for $\nu\in [3/5, 1]$ it is thus enough to prove the non-negativity of $b_2(x)$ for $x\in [-1/2, 3/10]$, or, equivalently, the non-negativity of
\begin{align*}\begin{split}
 g_2(y)&:= \frac14(1 -2y)^2\big(384 -5\pi^4 +(30\pi^4 -384\pi -96\pi^2 -32\pi^3)y \\&+(192\pi^2 +128\pi^3 -60\pi^4)y^2 +(40\pi^4 -128\pi^3)y^3\big)^2b_2(y -1/2)
\end{split}\end{align*}
for $y \in[0, 4/5]$. Expanding $g_2$ we obtain the representation
\begin{align}\label{g_2}
 g_2(y) =p_4(y)\sqrt{1 -y^2} +p_5(y)\sqrt{13 -4y^2}
\end{align}
with polynomials $p_4$ and $p_5$ of order $10$. We claim that
\begin{align*}
p_4(y) =\sum_{k =0}^{10}v_ky^k
\end{align*}
is positive for all $y\in [0, 4/5]$. Indeed, a positive lower bound on $p_4$ can be obtained by first replacing the positive coefficients $v_5$, $v_7$ and $v_9$ by zero, then for $k =3, 4, 6, 8, 10$ (for which $v_k$ are negative) estimating $v_ky^{k}$ from below by $(4/5)^{k -2}v_ky^2$, and estimating $v_1y$ from below by $4v_1/5$ (since $v_1$ is negative). As a result we obtain the estimate
\begin{align*}
 p_4(y)\geqslant \widetilde v_0 +\widetilde v_2y^2 >0 \quad\text{for }y\in[0, 4/5]
\end{align*}
with positive $\widetilde v_0$ and $\widetilde v_2$.
Thus to obtain the non-negativity of \eqref{g_2} on $[0, 4/5]$ it suffices to prove positivity of the polynomial
\begin{align*}
 p_6(y) :=\frac{p_4^2(y)(1 -y^2) -p_5^2(y)(13 -4y^2)}{81(1 -2y)^4} =\sum_{k =0}^{16}w_ky^k
\end{align*}
of degree $16$ on this interval. To get a positive lower bound on $p_6(y)$ for $y\in[0, 4/5]$ we can replace the positive coefficients at $w_{13}$ and $w_{15}$ by zero. For all other $k\geqslant 7$ we have $w_k <0$ and estimate $w_ky^k \geqslant (4/5)^{k -6}w_ky^6$. For $k\in \{1, \dots, 5\}$ we have $w_k >0$ and estimate $w_ky^k \geqslant w_ky^6$. As a result we obtain $p_6(y) \geqslant w_0 +\widetilde w_6y^6$ with $w_0, \widetilde w_6 >0$.

\paragraph{Non-negativity of $c_4(\nu)$ for $\nu\in[0, 1]$.} By the non-negativity of $c_2$ we obtain the lower bound
\begin{align}\label{c4 bound1}
 16\nu\Big(-224\nu + \frac{2401\nu}{9 + 4 \nu^2} + 16 \sqrt{9 + 4 \nu^2}(1 -\eta_\nu^2)\Big)
\end{align}
on \eqref{c4}. By Lemma \ref{l: eta_nu} the last term is non-negative and monotonously growing. Since the sum of the first two terms is non-negative for $\nu\in\big[0, \sqrt{55/128}\big]$, it is enough to establish the positivity of \ref{c4 bound1} for $\nu\in\big[ \sqrt{55/128}, 1\big]$. By concavity we can estimate
\begin{align}\label{c4 first term}
 -224\nu + \frac{2401\nu}{9 + 4 \nu^2} \geqslant \frac7{26} (16 + \sqrt{110})(\sqrt{110} -16\nu)
\end{align}
(the right hand side is a linear interpolation between the values at $\nu =\sqrt{55/128}$ and $\nu =1$) and, by convexity,
\begin{align}\label{c4 second term}
 \sqrt{9 +4\nu^2}\geqslant \frac{9 + 4\nu}{\sqrt{13}}
\end{align}
(the right hand side is the tangent line at $\nu =1$). Substituting \eqref{c4 first term} and \eqref{c4 second term} into \eqref{c4 bound1} and using $\eta_\nu \leqslant \eta_{\sqrt{55/128}}$ for $\nu\in\big[\sqrt{55/128}, 1\big]$
we get a lower bound
\begin{align*}
 16\nu\Big(\frac7{26}\big(16 +\sqrt{110}\big)\big(\sqrt{110} -16\nu\big) +\frac{16(9 + 4\nu)}{\sqrt{13}}\big(1- \eta_{\sqrt{55/128}}^2\big)\Big),
\end{align*}
on \eqref{c4}, where the last factor is a decreasing linear function positive for $\nu\leqslant 1$.

\paragraph{Non-negativity of $c_6(\nu)$ for $\nu\in[0, 1]$.} We first observe that by the non-negativity of $c_2$ we have a lower bound
\begin{align*}
 c_6(\nu) \geqslant \frac{6272\nu(1- \eta_\nu^2)}{\sqrt{4\nu^2+9}} -1024\nu^2.
\end{align*}
Its non-negativity is equivalent to
\begin{align}\label{1-eta}
 1 -\eta_\nu^2 \geqslant (8/49)\nu\sqrt{4\nu^2 +9},
\end{align}
where by Lemma \ref{l: eta_nu} both sides are monotonously growing. Due to
\begin{align*}
 1 -\eta_{1/2}^2 > (8/49)\sqrt{13}
\end{align*}
it remains to establish \eqref{1-eta} for $\nu\in [0, 1/2]$. For this we substitute \eqref{eta_nu} and the estimate \eqref{bound 1} to obtain that \eqref{1-eta} follows from
\begin{align}\label{c6 sufficient}
 (1280 +288\pi\nu^2 -72\pi^2\nu^4)\sqrt{4\nu^2 +9} \geqslant 441\pi^2\nu^3 -(1764\pi +3920)\nu.
\end{align}
Since the left hand side is clearly positive for $\nu\in [0, 1/2]$, \eqref{c6 sufficient} follows from the non-negativity of the polynomial
\begin{align*}\begin{split}
 p_7(\nu) &:=(1280 +288\pi\nu^2 -72\pi^2\nu^4)^2(4\nu^2 +9) - (441\pi^2\nu^2 -1764\pi -3920)^2\nu^2\\ &=:\sum_{k =0}^5 a_{2k}\nu^{2k}.
\end{split}\end{align*}
To obtain a lower bound on $p_7(\nu)$ for $\nu\in [0, 1/2]$ we use $a_{10} \geqslant0$ and for $a_8, a_6 <0$ apply the estimates $a_8\nu^8 \geqslant (1/2)^4a_8\nu^4$, $a_6\nu^6 \geqslant (1/2)^2a_6\nu^4$ obtaining
\begin{align*}
 p_7(\nu) \geqslant a_0 + a_2\nu^2 +\big(a_4 +(1/2)^4a_8 +(1/2)^2a_6\big)\nu^4,
\end{align*}
where the right hand side is positive.

\end{document}